%% file: main.tex
\documentclass[letterpaper,twocolumn,10pt]{article}
\usepackage{Style-Files/preamble}

\usepackage{graphicx} 
\usepackage{amsmath}
\usepackage{graphicx}
\usepackage{algorithm}
\usepackage[noend]{algpseudocode}
\usepackage{amsthm}
\usepackage{mathtools}

\newcommand{\natacha}[1]{{\color{green}\grumbler{Natacha}{#1}}}

\newcommand{\hein}[1]{{\color{blue}\grumbler{Hein} {#1}}}

\title{\sys{}: saving BFT through racing}

\newcommand{\inst}[1]{\textsuperscript{\ensuremath{#1}}}

\author{
{\rm Neil Giridharan\inst{\ast} \quad
Shubham Mishra\inst{\ast} \quad
Lorenzo Alvisi\inst{\dagger} \quad
Natacha Crooks\inst{\ast}}\\
{\rm Benjamin Marsh\inst{\ddagger} \quad
Hein Meling\inst{1} \quad
Kartik Nayak\inst{2} \quad
Grzegorz Prusak\inst{\ddagger}}\\
{\small
\inst{\ast}UC Berkeley \qquad
\inst{\dagger}Cornell University \qquad
\inst{\ddagger}Sei Labs \qquad
\inst{1}University of Stavanger \qquad
\inst{2}Duke University \qquad
}
}

\begin{document}
\maketitle
\pagestyle{empty}
\input{Abstract/abstract}
\input{Intro/alternate_intro3}

\input{Background/alternate_background3}
\input{Model/model}
\input{Overview/alternate_overview3}
\input{Protocol/ferrari_attempt2}
\input{Evaluation/eval-cr}
\input{Related-Work/related_work}

\input{Conclusion/conclusion}

\section*{Acknowledgments}
We thank our shepherd, and the
anonymous reviewers for their thorough and insightful
comments. This work was supported in part by the NSF grants CSR-CORE 2106842, NSF CAREER 2442542, a Sui
Research Award as well as gifts from Accenture,
AMD, Anyscale, Google, IBM, Intel, Mohamed Bin Zayed
University of Artificial Intelligence, Samsung SDS, SAP, and VMware.

\bibliographystyle{ACM-Reference-Format}
\bibliography{References/references}

\tr{
\newpage
\appendix
\input{Protocol/retry_protocol_appendix}
\input{Proofs/new-proofs}
\input{Tech-Discussion/tech-discussion}
}

\end{document}

%% file: Abstract/abstract.tex
\begin{abstract} 
Today's practical Byzantine Fault Tolerant (BFT) state machine replication deployments are vulnerable to slowdowns~\cite{visigoth, copilot}. The main culprit is timeouts. Aggressive timeouts spuriously trigger expensive leader changes, while conservative timeouts leave the system idle and let slowdowns severely inflate latency. Two main alternatives exist: hedging~\cite{tennage2023quepaxa, parbft}, which improves recovery from slow leaders but still incurs a time-based hedging delay, and cooperative asynchronous protocols~\cite{vaba,honeybadger}, which recover quickly from slowdowns but suffer from high common-case latency and low throughput. This paper presents \sys: a BFT state machine replication protocol that sidesteps this trade-off through \textit{protocol-rigged races}, where replicas, rather than race against the clock, race against each other by executing protocol steps. This enables \sys{} to achieve high throughput and low latency comparable to state-of-the-art timeout-based BFT, while matching the robustness of cooperative approaches.


\end{abstract}

%% file: Intro/alternate_intro3.tex
\section{Introduction}

This paper presents \sys, a Byzantine Fault Tolerant (BFT) state machine replication (SMR) protocol that achieves low latency, high throughput, and robustness to slowdowns (including failures and network blips).

Distributed trust systems allow multiple mutually distrustful parties to secure data even when a subset of them are compromised or misbehave. These systems have been used for secure key recovery~\cite{signal}, code transparency~\cite{microsoft-ccf}, securing land records~\cite{land-records}, and stablecoins~\cite{digital-euro}. Signal~\cite{signal}, for instance, distributes its key recovery logic across three clouds for better security. At the core of these systems lies a BFT SMR protocol, which implements the abstraction of a totally ordered log. This log can then be materialized into application state that is consistent across all parties, even when a subset crashes or misbehaves. These protocols must \one sustain high throughput, \two ensure low latency, and \three remain robust to actual failures and network events.

Much of the focus on robustness has been on tolerating discrete events: a replica failure~\cite{aardvark} or network events in which messages are dropped~\cite{autobahn}. Much less attention has been given to slowdowns~\cite{visigoth}. Slowdowns are events such as network misconfigurations, partial hardware failures, etc.~\cite{partial-partition, gray-failure, cloudflare-byz} that cause a replica to respond more slowly than normal. While slowdowns of course encapsulate machine crashes or network failures (a crashed node is infinitely slow), they also include events that cause latency spikes over longer time periods, without necessarily triggering timeouts.

These events are of major concern to real distributed system deployments. 
Sei~\cite{marsh2025seigiga}, a blockchain company running the Autobahn~\cite{autobahn} protocol, as well as Etcd and Neo4j (both Raft-based production datastores) report slowdowns of over seven seconds as a result of data synchronization issues, I/O contention, and garbage collection~\cite{marsh2025seigiga,neo4j,delayedheartbeat}.





Although there exist crash-fault-tolerant protocols robust to slowdowns~\cite{copilot}, most  existing state-of-the-art BFT protocols are vulnerable to this issue. 
Even systems like Autobahn~\cite{autobahn}, which explicitly target recovering from network events, are vulnerable. Because they rely on a leader to drive protocol execution, tail latency can spike when that leader is slow.
Techniques such as monitoring leader performance~\cite{aardvark} can detect some slowdowns, but are not comprehensive and offer only partial protection~\cite{copilot}.

This paper argues that \textit{timeout-triggered fault-detection is fundamentally the wrong mechanism for detecting slowdowns in an SMR protocol}, for two reasons.
First, it is a \textit{destructive} process. If timeouts are configured too aggressively, the system might unfairly accuse a leader of being down, unnecessarily triggering an expensive recovery process to elect a new leader. 
Moreover, trying to elect a new leader while the old leader is still alive can lead to concurrent proposals from both replicas. These proposals will interfere with each other, precluding either from making progress. Second, timeout-triggered fault detection is a \textit{blocking} process. Setting the timeout too conservatively causes the system to wait idle for long periods of time prior to detecting a failure. 
This is especially problematic as production deployments today recommend setting timeouts to at least ten times the network round-trip time \cite{kubstorage}. Timeouts are thus almost always configured to be on the order of seconds for geo-distributed deployments, with some production deployment choosing upwards of 30 seconds~\cite{baudet2019state,mystenprivate}! As a result, slowdowns often remain undetected, while truly failed leaders cause the system to stall for seconds~\cite{autobahn}.

To gracefully handle slowdowns, we need a new fault-detection mechanism that is instead \textit{guaranteed to make forward progress}. It should 1) be \textit{cooperative} rather than \textit{destructive}: it should not destructively interfere with ongoing proposal invocations from existing leaders; and 2) be \textit{productive}: it should contribute useful work towards committing operations.
Together, these properties minimize the system's recovery time when there is a slow or failed leader.

Hedging~\cite{tennage2023quepaxa} recognizes the issues of timeout-based approaches and offers a promising avenue for better handling of slowdown events. It replaces the traditional competitive leader-based proposal process with a novel cooperative multi-proposer approach. Each replica proposes in a staggered fashion, according to a known delay schedule.  Concurrent proposals contribute to faster termination when there are slowdowns or faults.  Consequently, hedging delays can be set much smaller than timeouts. Hedging, unfortunately, still relies on time-based delays. If the first proposer (leader) is slow, then the next replica must stall for the full hedging delay before it can propose. As such, hedging remains a blocking process that is not productive: stalling  does not contribute to the protocol's forward progress. 

Fundamentally, timeout-based approaches detect slowdowns by having the leader \textit{race against the clock}. The leader is given a fixed deadline to finish its work. If it fails to finish before that deadline, it loses the race and is deemed faulty or slow. Hedging, on the other hand, replaces this fixed deadline with a \textit{rigged race}. Instead of racing against the clock, the leader races against other replicas that are also proposing. Time is used only to give the leader a head start, ensuring that when the leader is not slow, it comfortably wins the race.


In this work, we argue that \textit{rigged racing} is indeed the correct mechanism for handling slowdowns.  But racing against the clock or using time to rig the race 1) causes replicas to wait idle 2) requires precise knowledge of what value to set the clock to. Replicas should instead \textit{race against each other}, using protocol messages to rig the race. We purposely design the leader's proposing protocol to take fewer steps than the other replicas' proposing protocol, so that the leader finishes proposing first in the absence of slowdowns. Unlike time-based rigging, non-leader replicas stay productive as they do work that can be used to commit an operation \textit{if they happen to win the race}. 

To instantiate this idea, we propose \sys{}, a new consensus protocol that achieves low latency, high throughput, and better robustness to slowdowns through \textit{protocol rigged races}. In the absence of a slowdown, the leader always wins the race, matching the latency of traditional leader-based BFT. In the presence of slowdowns, the leader will lose the race, but the other replicas will quickly recover using the work they did during the race.

Naturally, realizing this in practice requires addressing two key challenges: \one replicas may disagree on whether the leader won the race, and \two replicas may commit at different stages, with some committing the leader's value during the race and others committing only after the race's dust settles. 

Our results are promising. \sys{} matches the common-case performance of Autobahn~\cite{autobahn}, a state-of-the-art timeout-based BFT protocol, while outperforming it under slowdowns, achieving 1.6–3.0x lower peak latency for 1–2 second slowdowns and up to 10.8x lower latency for more severe cases. Compared to ParBFT2~\cite{parbft}, a leading hedging-based BFT protocol, \sys{} delivers 1.3x higher throughput and 1.9x lower common-case latency, along with 1.7–3.1x lower peak latency during slowdowns. 


This paper makes three core contributions:
\begin{itemize}[topsep=0pt]
        \item We identify the two necessary conditions for slowdown-detection mechanisms to be precise and efficient: being cooperative and productive. 
        \item We introduce a new slowdown-detection mechanism, protocol-rigged racing, the first to satisfy both conditions.
        \item We implement and evaluate \sys{}, a novel BFT protocol that uses protocol-rigged racing to offer high throughput, low latency, and robustness to slowdowns. \sys{} is deployed in production at a leading distributed trust company. 
\end{itemize}

%% file: Background/alternate_background3.tex
\section{Background}

We first survey the slowdowns observed in real deployments, and then examine the mechanisms used to detect them.

\subsection{Slowdowns in practice}

In production, distributed systems often experience \emph{slowdowns}: periods when a replica remains correct and available, yet its service rate falls below its usual baseline. Even brief slowdowns can cause end-to-end latency spikes and trigger timeouts.

Slowdowns arise from a mix of hardware, software, and network effects. Fail-slow behavior is prevalent across the hardware stack: disks may suffer throughput degradation from vibrations, memory cards from loose connections, CPUs from power loss. Even NICs can drop packets because of  firmware or driver bugs~\cite{fail-slow-scale, nvm-fail-slow}. Software slowdowns arise because of I/O contention, slow data syncs, or garbage collection pauses. I/O contention in etcd, for instance, can induce long \texttt{fsync} latencies~\cite{delayedheartbeat}; since proposals must be persisted to disk, these stalls directly translate into large latency spikes. Data synchronization likewise causes sufficiently long latency spikes in HBase that developers have added a dedicated feature to detect slow syncs~\cite{hbaseslowsync}. Sui~\cite{sui}, a production blockchain system, similarly observes that slow data synchronization can trigger slowdowns lasting up to 5 seconds~\cite{mystenprivate}. Elasticsearch and Neo4j both report that rapid heap allocation can trigger stop-the-world garbage collection pauses lasting seconds to minutes, inadvertently causing leader election timeouts~\cite{elasticsearchgc, neo4j}. Finally, network pathologies, including misconfigurations and flaky links have been shown to cause partial network partitions in several production data stores~\cite{towardgeneric}.

\subsection{Slowdown detection mechanisms}
Slowdowns are managed in a variety of ways today.


\par \textbf{Timeouts.} The most common way to detect slowdowns is to use timeouts. Timeouts effectively make the leader race against the clock to finish committing its proposal. If the leader loses, the timeout fires, triggering a leader election process to replace it.

Tuning timeouts is notoriously difficult~\cite{tennage2023quepaxa}. It requires significant administrative effort, and a poorly chosen value can degrade performance or compromise availability. Aggressive timeouts risk frequent leader elections, in which new leaders override the progress of old leaders, hurting liveness. Conservative timeouts give leaders plenty of time to win the race, but cause the system to stall in the process. This is especially problematic when the leader fails, as other replicas must wait out the full timeout to detect that failure. Moreover, timeouts are typically set statically, which makes them adapt poorly to changing conditions, such as network latency fluctuations or replicas that become gradually slower over time. Taken together, these limitations make clear that timeouts are an overly coarse heuristic for detecting slowdowns.


Production systems recommend setting timeouts to be at least ten times the network RTT, though in practice they are even more conservative. Default timeout values include 2 seconds for Sei~\cite{marsh2025seigiga}, 5~s for Microsoft CCF~\cite{ccf-timeout}, 7~s for Neo4j~\cite{neo4j}, and 30~s for Diem~\cite{diem2021,mystenprivate}.

\par \textbf{Cooperative Proposing.} Timeout-based protocols make progress only if the leader ``beats the clock". When it cannot, the system stalls until the replicas elect a new leader. Effectively, these protocols are built around the assumption that, in the typical case, the leader usually wins this race~\cite{aardvark, good-case}.

Asynchronous consensus protocols~\cite{honeybadger,vaba,2pac,liu2023flexible} instead start from a different premise: \textit{all} leaders may consistently lose the race against the clock. This can be a result of poorly configured timeouts or because the network is controlled by an adversary that selectively delays the leader's messages. Thus, to avoid blocking on a single leader, asynchronous protocols let all replicas propose cooperatively. During slowdowns, this approach achieves low latency and remains robust up to $f$ slow replicas, since $n-f$ correct replicas are guaranteed to eventually respond. However, this approach is more expensive in the common case, since extra coordination is required to handle multiple proposers. Asynchronous protocols thus have significantly higher latency in the absence of slowdowns. 

\par \textbf{Hedging.} The most recent slowdown detection mechanism, hedging~\cite{tennage2023quepaxa,parbft}, takes a middle-ground between the timeout and cooperative proposing approaches. With hedging, all replicas cooperatively propose (like in asynchronous protocols), but they do so in a staggered fashion according to a known delay schedule. This creates a race between the leader (the first proposer) and the remaining replicas to finish proposing. If the leader wins the race, the latency matches that of traditional leader-based BFT. Otherwise, the other proposers cooperatively work together to terminate the protocol. Hedging preserves the efficiency of timeout-based protocols in the absence of slowdowns. Unfortunately, during a slowdown, recovery is unnecessarily delayed until after the leader's head start. Furthermore, extra coordination is needed to reconcile the leader-based and cooperative mechanisms. During slowdowns, ParBFT2~\cite{parbft}, the state-of-the-art hedging BFT protocol, has a consensus latency of $22$ message delays on top of the hedging delay, compared to only $10$ message delays for SMVBA~\cite{speeding-dumbo}, the state-of-the-art cooperative BFT protocol.

\par \textbf{Toward Rigged Racing.} 
In short, current detection mechanisms are inadequate. Timeouts are too slow to react to slowdowns, as the system stalls until the timeout triggers. Cooperative approaches are not practical, as they unnecessarily pay the cost of higher latency always. Hedging, while an improvement, still stalls until the hedging delay is over.

A better slowdown detection mechanism should meet three requirements. First, it must be \textit{cooperative}: concurrent proposals should not destructively interfere with each other. Second, it must be \textit{productive}: the work done during detection should directly accelerate commitment once the slowdown is detected. Third, it must detect real slowdowns quickly without mislabeling a healthy replica as slow. 

These requirements rule out timeouts and hedging, since both rely on time-based delays that fundamentally do not contribute to useful protocol progress. A more promising direction instead quantifies slowness by measuring how quickly replicas advance through protocol steps that replicas must anyway execute. Truly slow replicas will consistently lag behind. However, raw relative progress is noisy: routine jitter can briefly make a healthy replica look slow. To satisfy the third requirement, we need a mechanism that separates genuine delay from benign variation. 
Our solution, \textit{protocol-rigged racing}, meets all three requirements by racing replicas while carefully adding sufficient bias to prevent false accusations.

%% file: Model/model.tex
\section{System Model}

\sys{} adopts the assumptions of prior BFT protocols. We assume $n=3f+1$ replicas, of which at most $f$ are faulty, as well as a PKI for digital signatures, and a trusted setup for threshold signatures. We consider a replica correct if it adheres to the protocol specification; a replica that deviates is considered faulty (or Byzantine). We make no assumptions about the number of faulty clients. We assume the existence of a strong, yet static adversary, that can corrupt and coordinate all faulty participants’ actions but cannot break standard cryptographic primitives. Replicas communicate through reliable, authenticated, point-to-point channels; we write $\langle M \rangle_p$ for messages signed by replica $p$. We assume all signatures and quorums are validated and omit the check in protocol descriptions.
\sys{} operates under the asynchronous network model: it makes
no synchrony assumptions for safety or liveness. 


%% file: Overview/alternate_overview3.tex
\section{\sys: Overview}

\sys{} seeks to achieve two goals: \one match the high performance of timeout-based BFT protocols in the absence of slowdowns, and \two provide stronger robustness to slowdowns than both hedging-based and timeout-based BFT protocols.

Timeout-based BFT protocols enjoy low latency by relying heavily on a designated leader: they assume the leader will beat the clock most of the time. When that happens, the leader becomes the sole replica authorized to propose new operations, and agreement can be reached quickly, with only three message delays~\cite{castro1999pbft}. Unfortunately, these protocols pay a steep price when the leader loses the race. Hedging-based protocols improve robustness by having all replicas propose concurrently when the leader falls behind, but they still pay both the hedging delay and additional coordination needed to reconcile all proposals. \sys{} overcomes these limitations and achieves both objectives through a new technique called \textit{protocol-rigged racing}.


In \sys{}, a designated leader drives commitment of its proposal in its \textit{proposal lane}. Additionally, each replica has a proposal lane of its own, where it proposes operations and votes in the lanes of others.

\sys{} consists of two components: a race, in which replicas try to beat the leader (but assume that they won't), and a recovery path in which the replicas try to commit their own proposal if they believe that the leader lost.
During the race, the leader runs a faster protocol that ensures that it will win if it is not slow. If the leader loses, replicas do not have to turn back and start over. Instead, the steps taken to beat the leader are directly useful towards committing.




\par \textbf{Race.}  Replicas race with the leader by submitting their own proposals. Rather than having replicas perform meaningless busy work during the race, \sys{}'s key insight is to have them execute required protocol operations during this time, shortening the amount of work left post-race.
\textit{All} consensus protocols fundamentally consist of two phases: a \textit{non-equivocation} phase and a \textit{persistence} phase. Non-equivocation guarantees that malicious replicas cannot cause conflicting operations to commit; persistence ensures committed operations survive. \sys{} uses the non-equivocation phase as the race itself! Both leaders and replicas need it \textit{anyways} to commit. 

To ensure a non-slow leader will win the race, the race is rigged in the leader's favor in two ways. First, the leader completes the two-step, quadratic message exchange of PBFT's non-equivocation phase while the other replicas instead perform a slower three-way linear message exchange. Second, the leader does not have to finish first to win; it only has to beat the cutoff. This cutoff is chosen to maximize the leader's chances of winning, while still allowing replicas to recover quickly if the leader is slow.
When the cutoff is reached locally, the race outcome is determined. The key challenge is that replicas may disagree on the race outcome. One replica might observe the leader winning, while another sees the leader losing. Despite this divergence, \sys{} ensures agreement across all participants.


\par \textbf{Recovery Path.} If a replica failed to commit the leader's proposal, it moves to the recovery path. Unfortunately, there is no longer a dedicated leader from which to choose a proposal. Instead, all replicas potentially have formed a distinct non-equivocation certificate for their lane. To make matters worse, the leader lane could have committed with only some replicas knowing. The recovery path addresses each issue in turn and proceeds in three steps. First, a replica chooses a value to recover for its lane. Second, it persists this value. And third, a lane is randomly selected to be committed.


\par \textbf{Step 1.} Replicas must recover any potential proposal that \textit{could} have committed from the leader. If a leader proposal could have committed, replicas must select that value as input to the recovery path. If no leader proposal was committed, replicas recover their proposal from the race to ensure that the progress made during the race was productive.

\par \textbf{Step 2.} Replicas then persist the recovered value, so that it is safe to commit it. This step has two parts. The first is an optional \textit{race exclusion phase} that addresses a subtle issue. If a leader non-equivocation certificate forms but the leader does not successfully commit, a malicious replica may successfully obtain both a non-equivocation certificate from the leader and for its own lane. It cannot equivocate on any proposal, but it still can equivocate between the leader's proposal and its own lane proposal. 
Thankfully, we find that, in most cases, the leader is slow enough that replicas can directly generate a proof that the leader's non-equivocation certificate \textit{cannot exist}. Most of the time, replicas thus safely omit this phase. Only when they cannot generate this proof must they proceed to run this additional exclusion phase to guarantee uniqueness between the leader's non-equivocation certificate and the replica's. The second part is the \textit{persistence phase}, which performs the same role as in traditional BFT. It guarantees that if the lane proposal is committed, then it survives despite failures.

\par \textbf{Step 3.} The third and final step is that a lane is \textit{randomly} selected for commitment. Note that this election process is run \textit{after} each replica has completed all the steps required to commit a proposal in its lane. This is by design: if the winner were chosen deterministically, a network adversary, could intentionally delay messages necessary for the anointed winner to commit the operation, defeating the point of not having a leader. If the randomly selected lane does not have a proposal ready for commit, the protocol simply retries the recovery path. 

\par \textbf{Performance Summary} Using these techniques, \sys{} achieves a latency of three message delays under normal conditions, matching the performance of latency-optimal PBFT~\cite{castro1999pbft}. When the leader experiences a slowdown, the expected latency is comparable with state-of-the-art pessimistic (asynchronous) protocols~\cite{speeding-dumbo, vaba}: $9.5$ message delays if the leader fails to propose during the race, and $10.5$ message delays otherwise. These results meet our goal of delivering state-of-the-art latency in both normal and slowdown scenarios.


%% file: Protocol/ferrari_attempt2.tex
\section{\sys{}: The Protocol}
We now describe the protocol in more detail. We first focus on how the protocol works for a single slot, before expanding it to support efficient agreement across multiple slots.
For ease of exposition, we omit the standard checks performed by all replicas---{\em i.e.}, whether messages are well-formed and from the correct view, and signatures are valid.
Replica state is summarized in Algorithm~\ref{alg:consensus-state}.
We defer formal proofs of correctness to \tr{Appendix \S\ref{s:proofs} (Theorems~\ref{thm:safety} and~\ref{thm:liveness})}{our supplementary material}.

Execution in Ambulance is organized into a sequence of views. Within a view, replicas try to commit a proposal. The protocol consists of two components: a \emph{race}, in which replicas determine whether the leader is fast enough to commit, and a \emph{recovery path}, in which they try to commit their own proposal if the leader is deemed slow. As previously stated, each replica in \sys{} has its own proposal lane. Replicas drive commitment for their own lane, but additionally vote in other replicas' proposal lanes. Note that the leader runs both a leader lane and a replica lane. 

\begin{figure*}[h!]
\centering
\begin{minipage}{.5\textwidth}
  \centering
  \includegraphics[width=\linewidth]{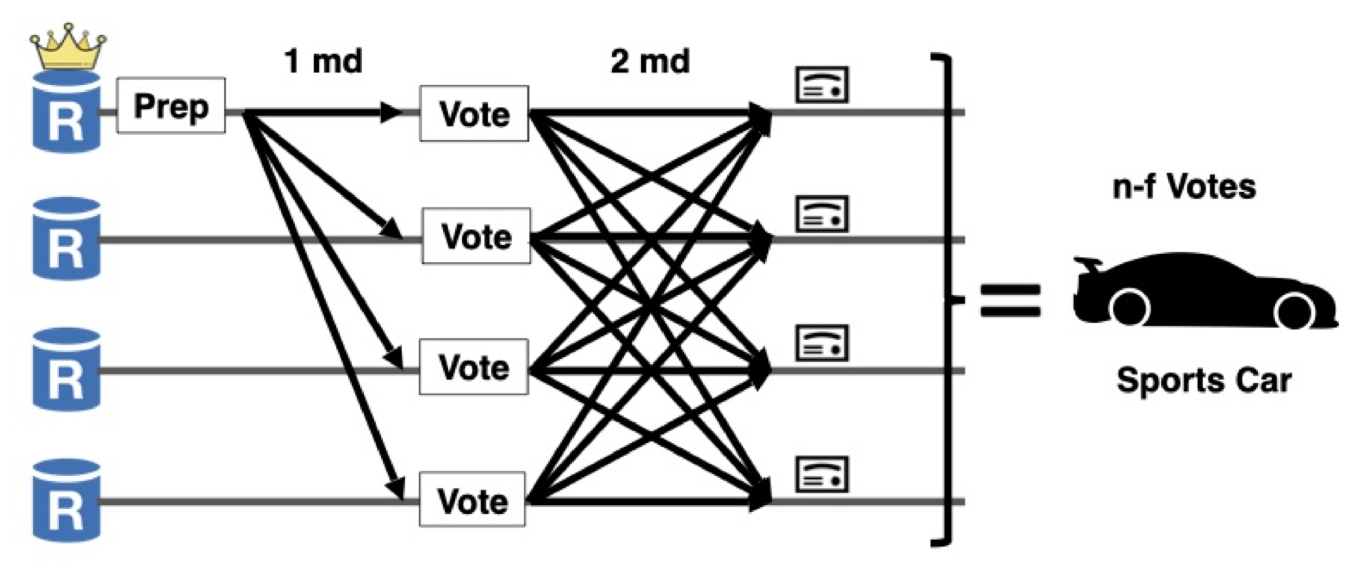}
  \caption{Leader lane protocol pattern.}
  \label{fig:leader-pattern}
\end{minipage}%
\begin{minipage}{.5\textwidth}
  \centering
  \includegraphics[width=\linewidth]{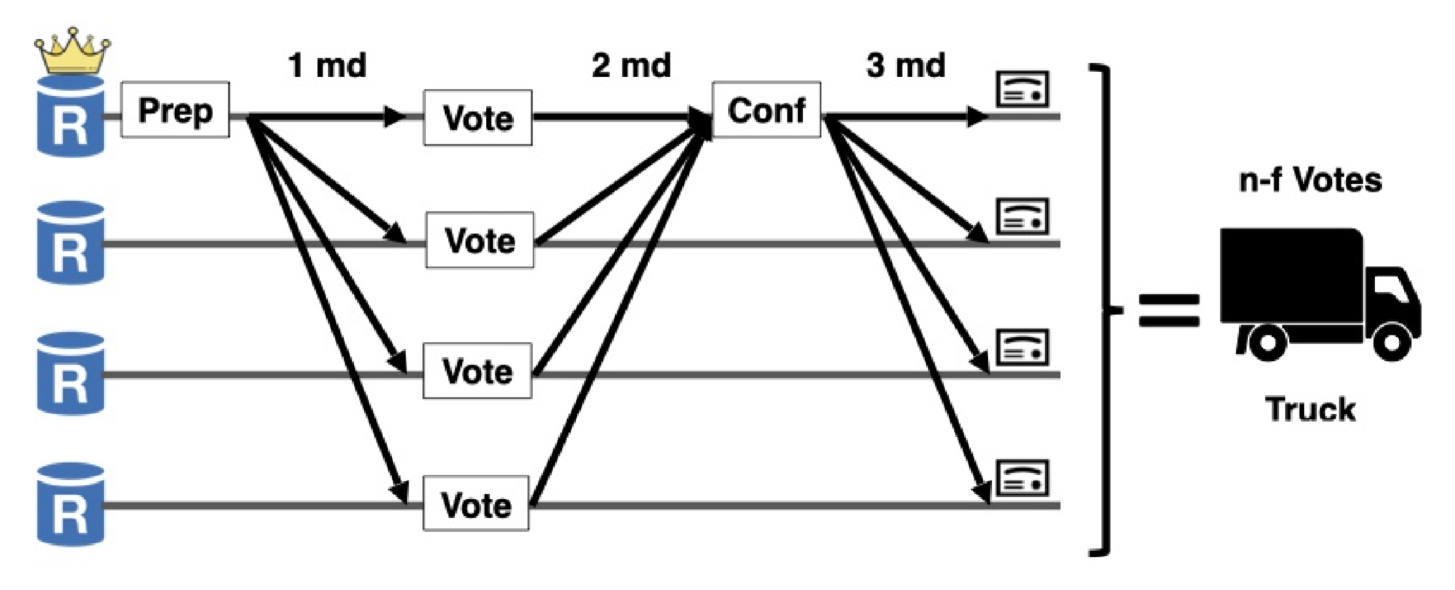}
  \caption{Replica lane protocol pattern.}
  \label{fig:replica-pattern}
\end{minipage}
\vskip -5pt
\end{figure*}

\vskip -0pt
\begin{algorithm}
\caption{Replica state cheat sheet for a given slot $s$.}\label{alg:consensus-state}
\begin{algorithmic}[1]
\State{\textit{v}}\hfill$\triangleright$ replica's current view
\State{\leaderpropname}\hfill$\triangleright$ \fastprepare from the leader or \noleaderprop{}
\State{\lockcert}\hfill$\triangleright$ \lockcertname or \noleaderlock{}
\State{\slowcert}\hfill$\triangleright$ \slowcertname
\State{\ulostcert}\hfill$\triangleright$ \ulostcertname from race
\end{algorithmic}
\end{algorithm}
\vskip -0pt

\subsection{The Race}\label{sec:fast-path}
The protocol begins with a race whose purpose is two-fold: \one correctly and quickly identify whether the leader is slow and \two allow a non-slow leader to commit.

To achieve this, the leader and replicas simultaneously issue proposals in their own lanes and race to form non-equivocation certificates for their lane. Such certificates guarantee that, in any view, only a single value can be proposed in the corresponding lane.
We call the leader's non-equivocation certificate a \emph{\lockcertname}, and a replica's non-equivocation certificate a \emph{\slowcertname}\footnote{Sports cars are optimized for speed but are fragile, whereas trucks are slower but more reliable.}. Non-equivocation is a necessary step in essentially all BFT protocols, so replicas make progress toward commitment while participating in the race.


We ensure that leaders progressing normally are not accidentally identified as  slow by rigging the race in the leader's favor in two ways. First, the leader runs a faster non-equivocation procedure that produces a \lockcertname in only two message delays, whereas each replica lane requires three message delays to produce a \slowcertname. This structural asymmetry naturally biases the race toward a non-slow leader. Second, the leader does not need to finish first overall to avoid being detected as slow. It only needs to finish prior to a \emph{cutoff}, which occurs when a quorum of \slowcertnameplural has formed. If a \lockcertname forms prior to the cutoff, the leader wins the race and replicas continue toward committing the leader's value. Otherwise, the leader loses the race and replicas proceed to the recovery path.

\subsubsection{Leader Lane: \lockcertnamecaps}

\sys{} structures the leader lane to guarantee that a non-slow leader will beat the cutoff while still making useful progress toward commitment. The leader lane implements the standard non-equivocation step required for commitment with a two-step all-to-all communication pattern, enabling a non-slow leader to produce a \emph{\lockcertname} in only two message delays (illustrated in Figure~\ref{fig:leader-pattern}).

\protocol{\textbf{1: L} $\rightarrow$ \textbf{R}: Leader $L$ sends \fastprepare to all.}

The leader broadcasts a \fastprepare $\coloneqq \langle v, val \rangle_L$ message, where $v$ is the view number and $val$ is the value proposed by $L$.
\par \protocol{\textbf{2: R} $\rightarrow$ \textbf{R}: Replica $R$ processes \fastprepare and sends vote to all replicas.}
\vspace{1pt}
On receiving a \fastprepare message $prep$ from the leader, $R$ checks that $prep$ is valid and stores it in \leaderpropname $= prep$.
It then broadcasts a \fastprepvote $\coloneqq \langle v, dig=h(prep.val) \rangle_R$ message to all replicas, where $h$ is a cryptographic hash function.

\protocol{\textbf{3: R}: Replica $R$ assembles a \lockcertname.}

$R$ waits for at least $n-f=2f+1$ matching \fastprepvote{} messages and aggregates them into a \textit{\lockcertname}. It sets \lockcert $\coloneqq (v, dig, \{$\fastprepvote$\})$.

This certificate has two purposes. First, it guarantees non-equivocation---the uniqueness of the leader’s proposal within a view---since correct replicas vote only once and  any two quorums of size $n-f$  intersect in at least one correct replica. Second, as we are about to see, it serves to determine whether the leader beat the cutoff. $R$ stores this certificate in $\lockcert$ for potential use during recovery.


\subsubsection{Replica Lanes: \slowcertnamecaps}

Each replica lane similarly produces a non-equivocation certificate, called a \emph{\slowcertname}.
Unlike the leader lane, a replica lane uses a three-step linear communication pattern, so \slowcertnameplural form more slowly than a \lockcertname by design (illustrated in Figure~\ref{fig:replica-pattern}).

We describe the protocol from the perspective of a single proposer replica $P$, which drives one such lane to completion.
In practice, all replicas concurrently act as proposers in their own lanes, and they all contribute to each other's lanes. 

\protocol{\textbf{1: P} $\rightarrow$ \textbf{R}: Replica $P$ sends \slowprepare to all replicas.}

A proposer $P$ broadcasts a message \slowprepare $\coloneqq \langle v, val\rangle_P$, where $v$ is the view number and $val$ is the value proposed by $P$.

\par \protocol{\textbf{2: R} $\rightarrow$ \textbf{P}: Replica $R$ processes \slowprepare and sends vote to $P$.}
Upon receiving a \slowprepare message $prep$ from $P$, $R$ verifies its validity and, if valid, replies to $P$ with a
\slowprepvote $\coloneqq \langle v, dig=h(prep.val) \rangle_R$ message.

\protocol{\textbf{3a: P}: Replica $P$ assembles a \slowcertname.}

$P$ aggregates $n-f$ distinct matching \slowprepvote messages into a \slowcertname, and sets \slowcert $\coloneqq (v, dig, \{$\slowprepvote$\})$.

\protocol{\textbf{3b: P} $\rightarrow$ \textbf{R}: Replica $P$ broadcasts \slowcommit.}

After forming a \slowcertname for that lane, $P$ sends a
\slowcommit $\coloneqq \langle v, \slowcert \rangle_P$
message to all replicas.

\subsubsection{Race Cutoff}
The race ends locally at a replica once it has received $n-f$ \slowcertnameplural. This cutoff is chosen to give the leader as much room as possible to win the race.
Since up to $f$ replicas may be faulty, no correct replica can wait for more than $n-f$ \slowcertnameplural. At the same time, setting the cutoff below $n-f$ would shorten the race unnecessarily, increasing the chance that the leader is falsely detected as slow even when it is not. Once the cutoff is reached, the replica either continues toward committing the leader's proposal, if it observed the leader win the race, or enters the recovery path, if it observed the leader lose.

\protocol{\textbf{1: R:} Replica $R$ assembles a \ulostcertname.}

Upon receiving $n-f$ matching valid \slowcommit messages from distinct replicas for this proposal lane, $R$ assembles a
\ulostcert $\coloneqq \{$\slowcommit$\}$ certificate. This certifies that the $(n-f)$th certificate has formed and hence cutoff is reached.

After the cutoff, each replica is in one of two cases: either \one it received a \lockcertname prior to the cutoff (Figure~\ref{fig:beat}) or \two the leader missed the cutoff (Figure~\ref{fig:miss}).
As replicas can receive certificates in a different order, their belief about who won the race may diverge. The \lockcertname forming prior to cutoff indicates that the leader is making progress at the expected speed. If instead a cutoff forms first, the leader is behaving unexpectedly slowly. We describe each case in turn.

\par \textbf{Case 1: The leader beat the cutoff.}
Suppose replica $R$ receives the \lockcertname before $n-f$ \slowcertnameplural. Then $R$ concludes locally that the leader is behaving as expected, and therefore its proposal should be committed. The protocol then moves toward committing the leader's value.

\protocol{\textbf{1: R $\rightarrow$ R}: Replica $R$ broadcasts \fastcommit.}

$R$ broadcasts
\fastcommit $\coloneqq \langle v, dig \rangle_R$.
The \fastcommit message serves two purposes:
\one it tells other replicas that $R$ locally observed the leader beating the cutoff and that the leader's proposal should therefore be committed; and
\two it ensures that $R$ has persisted the leader's proposal for use in recovery, in case recovery begins before $R$ commits the leader proposal. This is equivalent to the commit phase of traditional BFT protocols.

\protocol{\textbf{2: R $\rightarrow$ R}: Replica $R$ assembles and forwards a \leadercommitname.}

Replica $R$ aggregates $n-f$ matching \fastcommit messages into a \leadercommitname. It then locally commits the leader's value and forwards the commit certificate to all replicas so that they can commit as well. Because the leader lane has a single designated leader, a quorum of \fastcommit messages suffices to commit the leader's proposal. Conflicting leader proposals cannot arise in the same view because the \slowcertname guarantees non-equivocation. Note that even if $R$ saw the leader missing the cutoff locally, it can still commit the leader's proposal if it receives a quorum of \fastcommit messages from other replicas.

\par \textbf{Case 2: The leader missed the cutoff.}
Suppose replica $R$ observes $n-f$ \slowcertnameplural before the \lockcertname, then it concludes the leader is slow. It therefore abandons participation in the leader lane and stops sending \fastprepvote{} or \fastcommit{} messages. However, it continues processing \fastcommit{} messages from other replicas, since those messages may still certify that the leader's proposal has committed.


\begin{figure*}[h!]
\centering
\begin{minipage}{.5\textwidth}
  \centering
  \includegraphics[width=\linewidth]{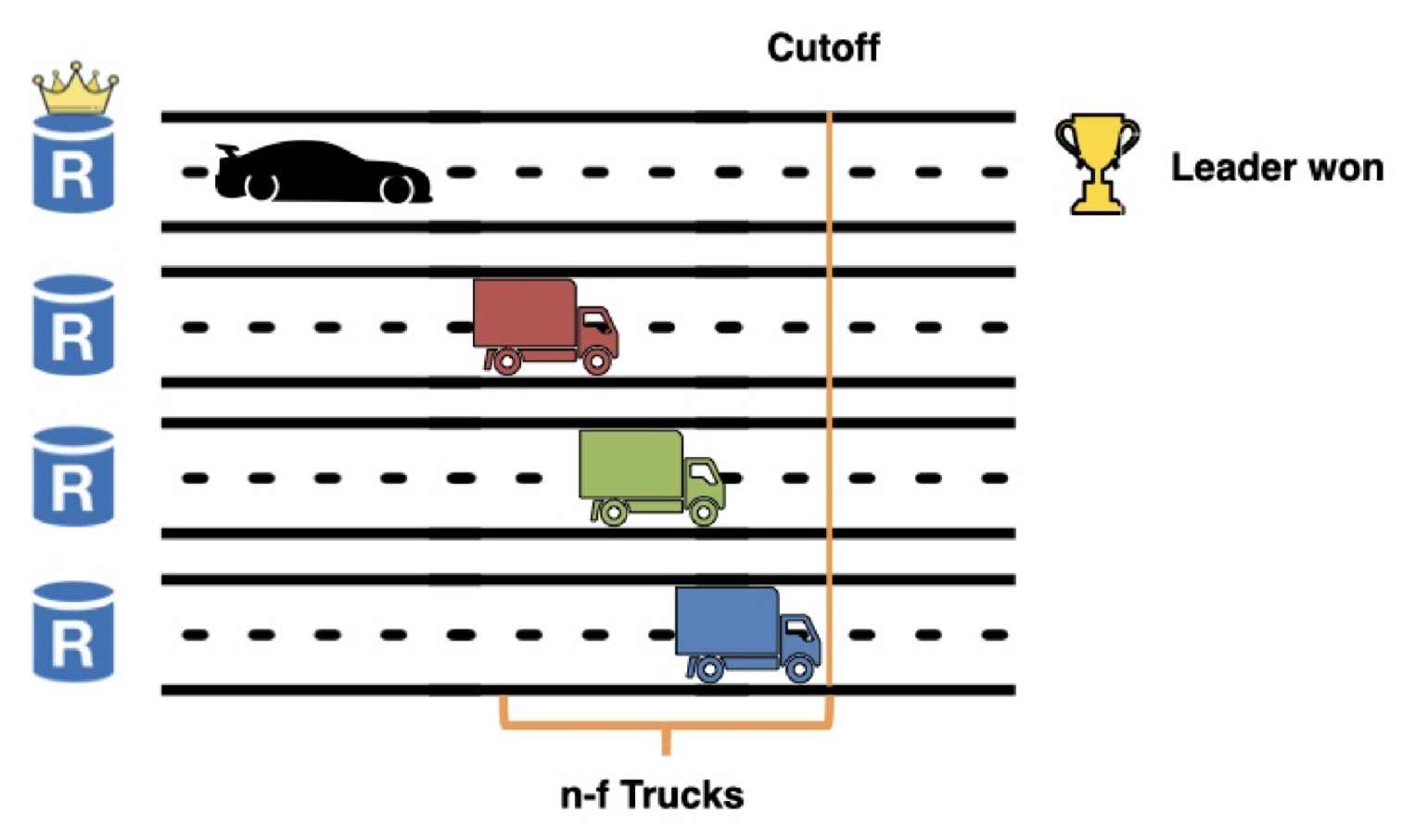}
  \caption{Leader beating the cutoff.}
  \label{fig:beat}
\end{minipage}%
\begin{minipage}{.5\textwidth}
  \centering
  \includegraphics[width=\linewidth]{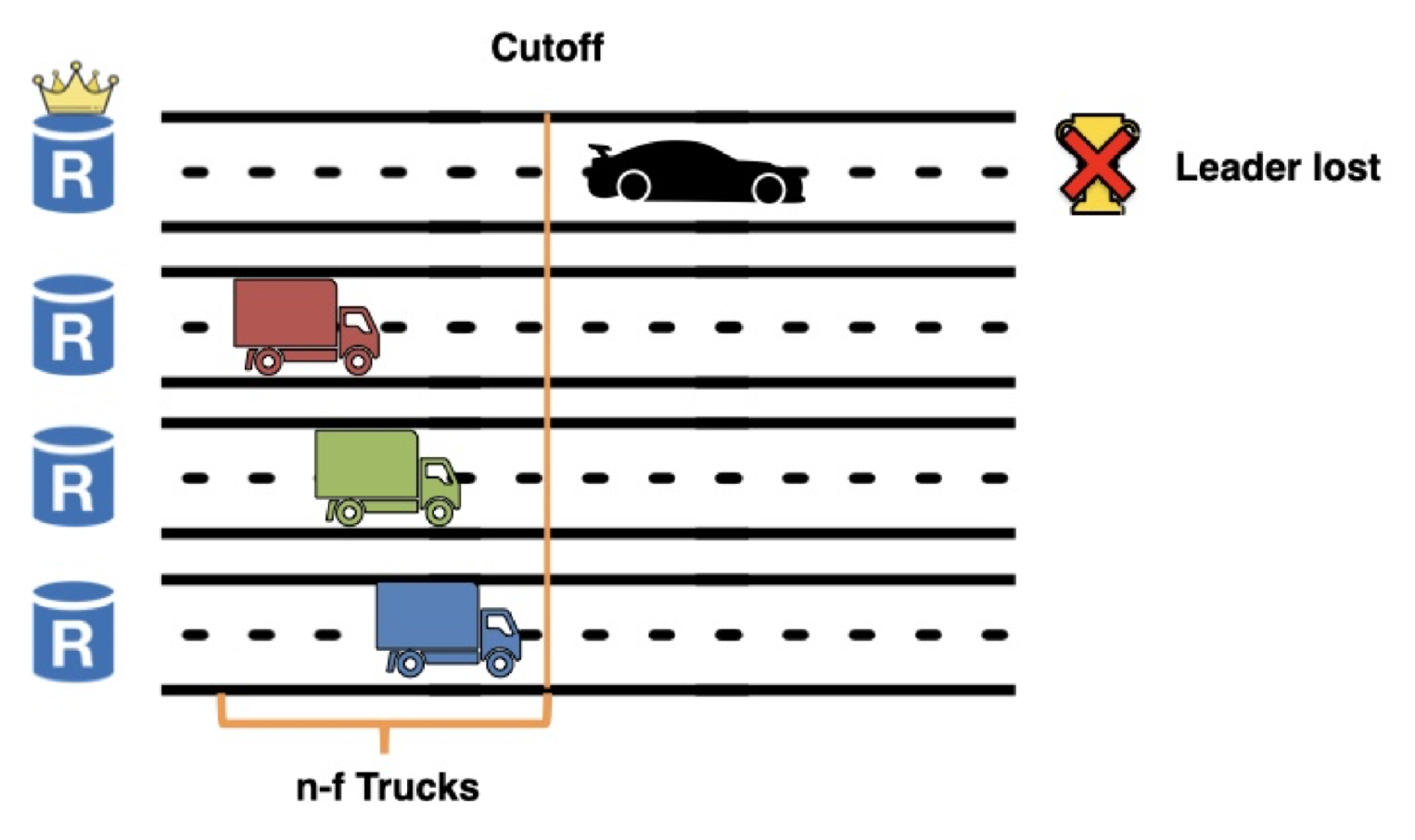}
  \caption{Leader missing the cutoff.}
  \label{fig:miss}
\end{minipage}
\vskip -5pt
\end{figure*}

\subsection{Recovery Path}
Once a replica locally reaches the cutoff, it enters the recovery path. The goal of the recovery path is to guarantee commitment of {\em some} proposal across all replicas. This may be the leader's proposal if some replica committed it during the race, or it may be some other replica's proposal.


In traditional BFT protocols, recovery proceeds in three steps: \one replicas elect a new leader, \two that leader chooses a value to recover, and \three the protocol then runs the usual non-equivocation and persistence phases for that value. However, this ordering of steps is problematic in our setting, where liveness does not rely on any synchrony assumptions. 
If the new leader is identified at the start of recovery, a network adversary can repeatedly target that replica and prevent progress. To avoid this vulnerability, we reverse the usual order of steps. Rather than electing a leader first, every replica initially acts as a potential leader by advancing its own recovery value through the consensus steps needed for commitment, namely non-equivocation and persistence. Leader election is deferred until a quorum has completed these steps, so the adversary learns whom to target only at the very end, when it is too late to interfere. This makes it likely that the election selects a replica from that quorum, at which point the protocol can terminate because that value has already been persisted. 

\vskip -0pt
\begin{algorithm}
\caption{Replica state cheat sheet for recovery path.}\label{alg:consensus-state-recovery}
\begin{algorithmic}[1]

\State{\certdata}\hfill$\triangleright$ \elected \preptwocert or \noelectcert{}
\State{\textit{qcs}}\hfill$\triangleright$ map: replica ID  $\rightarrow$ recovery path certificate

\State{\textit{finish}}\hfill$\triangleright$ set of \textsc{Finish} messages received

\end{algorithmic}
\end{algorithm} 
\vskip -0pt

\subsubsection{Step 1: Choosing a value to recover per lane}
In Step 1, each replica selects, for its lane, a value to recover based on the evidence gathered during the race. The key safety requirement is that if some value was already committed (via the \lockcertname), then replicas must recover that value.

Replicas first share with each other any messages they received from the leader during the race. This evidence will be used to determine what proposal should be recovered.

\protocol{\textbf{1: R} $\rightarrow$ \textbf{R}: Replica $R$ broadcasts \textsc{Status}.}
The replica broadcasts all state received from the leader during the race, including any proposal received from the leader ($\leaderpropname$). If a replica did not receive a \leaderpropname it sets $\leaderpropname=\noleaderprop{_R}$. Similarly, it sends any \lockcertname it received. If a replica did not receive a \lockcertname it sets $\lockcert=\noleaderlock{_R}$.

$R$ broadcasts this information to all other replicas as part of a \status $\coloneqq \langle \textit{\leaderpropname}, \textit{\lockcert} \rangle_R$ message.

\protocol{\textbf{2: } $P$ receives quorum of \status messages}

A replica $P$ selects a value to recover for its lane based on a quorum of $n-f$ \status messages. $P$ considers three cases: 

\protocol{\textbf{Case 1}: $P$ receives a \lockcertname.}
If a \lockcertname exists as part of any received \status message, then some replica may have formed a \leadercommitname (follows from quorum intersection). 
To preserve safety, $P$ therefore sets its recovery input to the leader’s value:
$\textit{r-input} = \lockcert$.

The next two cases handle the situation where $P$ does not receive a \lockcertname. In both cases, $P$ chooses the same recovery input; however, in the latter case $P$ can safely skip the \recoveryprepphase{} phase.

\protocol{\textbf{Case 2a}: $P$ receives $n-f$ \noleaderlock{} and a \fastprepare.}
Receiving $n-f$ \noleaderlock{} messages definitively indicates that no correct replica committed on the leader’s lane. This also follows from quorum intersection: if a \leadercommitname existed, $P$ would be guaranteed to receive a \lockcert, which it did not. In principle, $P$ could now safely recover any value, since no leader commit occurred. However, we require $P$ to recover its own \slowcertname from the race, thereby preserving the progress already achieved on its lane. 

In this case, $P$ sets its recovery input to $\textit{r-input} = \slowcert$, aggregates the $n-f$ signatures on \noleaderlock{} into a \nocommitcert{} $\textit{nc-cert}$, which proves that a leader commit could not have occurred, and proceeds to the next step.


\protocol{\textbf{Case 2b}: $P$ receives $n-f$ \noleaderlock{} and $n-f$ \noleaderprop{}.}
Receiving $n-f$ \noleaderprop{} messages definitively indicates that no \lockcertname exists. This follows from quorum intersection: if a \lockcertname existed, $P$ would be guaranteed to receive a corresponding \fastprepare, which it did not. Since there is no possible \lockcertname to equivocate with, $P$ can safely skip the \recoveryprepphase{} phase (see section~\ref{s:race-exclusion}): 
its own racing \slowcertname (\slowcert) already ensures non-equivocation within its lane. $P$ aggregates the $n-f$ signatures on \noleaderprop{} into a \nolockcert{} \nlcert, which proves that no \lockcertname exists, and sets its recovery input to
$\textit{r-input} = (\slowcert, \nlcert)$.

\subsubsection{Step 2: Persisting the recovered value per lane}
In Step 2, each replica persists its recovered value so that its value can be safely committed in Step 3. 
This step has two parts: the aforementioned optional \recoveryprepphase{} phase followed by a persistence phase that mirrors the familiar two-phase structure of traditional BFT. The optional \recoveryprepphase{} phase is needed because, although the racing path already guarantees non-equivocation within each lane, a Byzantine replica could still try to equivocate between a possible \lockcertname and its own \slowcertname. 
The \recoveryprepphase{} phase prevents this. The persistence phase then plays the usual role from traditional BFT: it records the recovered value so that subsequent protocol steps treat it as the lane's durable candidate.

\par \textbf{\recoveryprepphasecaps{} Phase.}\label{s:race-exclusion}
The \recoveryprepphase{} phase prevents a Byzantine replica from equivocating across lanes between its own \slowcertname and a possible \lockcertname. By the end of this phase, at most one of these values can advance in $P$’s lane: either $P$’s \slowcertname or the \lockcertname, but not both.

\protocol{\textbf{1: P} $\rightarrow$ \textbf{R}: Replica $P$ broadcasts \recoveryprepare.}
$P$ broadcasts \recoveryprepare $\coloneqq \langle v,\textit{r-input},\nccert \rangle_{P}$ where $v$ is the view, $\textit{r-input}$ is the selected recovery input (from step 1), and \nccert is the \nocommitcert{} formed from case 2a if it exists. 

\protocol{\textbf{2: R} $\rightarrow$ \textbf{P}: Replica $R$ processes \recoveryprepare and votes.}
$R$ checks whether the received \recoveryprepare is valid.
%
If so, $R$ replies by sending \recoveryprepvote $\coloneqq \langle v, dig=h(\textit{r-input}) \rangle_R$ back to the proposer.

\protocol{\textbf{3: P:} Replica $P$ assembles a quorum of \recoveryprepvote.}
$P$ waits for at least $n-f$ matching \recoveryprepvote messages, and aggregates these votes into a \preptwocertname, \preptwocert $\coloneqq (v,P,dig,\{$\recoveryprepvote$\}$) 
No two \preptwocertname can exist for the same lane because the necessary quorums intersect in at least one correct replica (which votes only once). For $P$'s lane, all replicas will either have a \preptwocertname for the leader's value or $P$'s value from the racing path but not both.

\par \textbf{Persistence Phase.}

\protocol{\textbf{1: P} $\rightarrow$ \textbf{R}: Replica $P$ broadcasts \recoveryconfirm.}
If $P$ could skip the \recoveryprepphase{} phase, $P$ sets $\preptwocert=\textit{r-input}$. This effectively upgrades its \slowcertname into a \preptwocertname. Otherwise, $P$ uses the \preptwocertname it formed in the \recoveryprepphase{} phase. $P$ broadcasts a message \recoveryconfirm $\coloneqq \langle v,\preptwocert \rangle_{P}$ where $v$ is the view, and $\preptwocert$ is the \preptwocertname $P$ is proposing.

\par \protocol{\textbf{2: R} $\rightarrow$ \textbf{P}:  $R$ receives and acknowledges \recoveryconfirm.}
$R$ checks whether the received \recoveryconfirm is valid.
If the validity checks pass, $R$ stores a copy of the \preptwocertname in $qcs[P]=$\preptwocert. Then, it replies by sending \recoveryconfirmvote $\coloneqq \langle v,dig=h($\preptwocert$) \rangle_R$ back to the proposer.

\protocol{\textbf{3: P} $\rightarrow$ \textbf{R}: Replica $P$ broadcasts \textsc{Finish}.}
$P$ waits for at least $n-f$ \recoveryconfirm messages with matching digests, and aggregates these votes into a \confirmcertname, \confirmcert $\coloneqq (v,dig,\{$\recoveryconfirmvote$\}$). It then forwards the \confirmcertname it formed in a $\textsc{Finish} \coloneqq \langle v,\confirmcert \rangle_P$ message.

\protocol{\textbf{4: R}: Replica $R$ receives quorum of \textsc{Finish} messages.}
For each \textsc{Finish} message, $R$ stores a copy of the \confirmcertname in $qcs[R_i]=\confirmcert$, where $R_i$ is the replica who proposed \confirmcert. Once $R$ receives $n-f$ \textsc{Finish} messages, it starts the lane election step.

\subsubsection{Step 3: Selecting a lane to commit} 
Unlike traditional BFT, which has a single dedicated proposal lane from the leader, \sys{} has multiple lanes, so it is not immediately clear which lane to commit. To mirror traditional BFT’s single-lane behavior, once $n-f$ lanes have completed the persistence phase, \sys{} randomly elects a single lane and checks whether to commit its proposal. If the elected lane has completed the persistence phase, the election succeeds and replicas commit the proposal in that lane. Otherwise, replicas move to the next view and rerun the recovery path. The lane election process must satisfy two properties: \one all correct replicas must agree on who won the election---this is necessary for safety; and  \two more surprisingly, the lucky winner cannot be predicted prior to running the protocol. The latter requirement originates from the potentially adversarial nature of slowdowns: if the "eventual" winner is known prior to its committing, a network adversary could systematically target it. \sys{} achieves both properties by leveraging threshold signatures to elect the winning lane.

\par \textbf{Lane Election.} Each replica $R$ produces a threshold-signature share on the view number, $v$, and broadcasts it. After collecting $2f+1$ valid shares for view $v$, a replica combines them into a single unique threshold signature $\sigma$ and computes the elected lane proposer as
$E \coloneqq h(\sigma)\bmod n$, where $h$ is a cryptographic hash function. Agreement follows because $\sigma$ is unique, so all correct replicas compute the same $E$; unpredictability follows because $E$ cannot be determined before the shares are revealed.

\par \textbf{Commit rule.} Replica $R$ commits if it completed the persistence phase for the elected lane ($qcs[E]$ contains a \confirmcertname). Otherwise, $R$ stops participating in view $v$, advances to the next view, and forwards $\sigma$ to help other replicas derive the same election outcome.

\subsection{Retry Protocol}
If a replica fails to commit in view $0$, it advances to the next view.  In all views $>0$, replicas execute only the recovery path (there is no racing path). The recovery path is identical to the recovery path of view $0$, with one small (but key) difference: the race-exclusion phase is always executed as there could be multiple valid lane inputs to recover. We defer the details of the retry protocol to Appendix \ref{s:retry}.

\subsection{Multi-Shot Agreement}
For multi-shot agreement, \sys{} constructs a totally ordered log of slots, where each slot runs an independent instance of the single-shot protocol. To avoid sequential delays, \sys{} pipelines these instances, so that a slot can begin without waiting on the previous slot to finish. We adopt the same pipelining strategy as Autobahn~\cite{autobahn}, where we start a new slot upon receiving a proposal from the previous slot, and cap the number of concurrent uncommitted slots to a pipelining parameter, $k$.

\subsection{Motorization}
Autobahn~\cite{autobahn} introduced the concept of motorization, which allows any consensus protocol to achieve higher throughput by adopting Autobahn's data layer. \sys{} can use Autobahn's data layer to parallelize the data dissemination process and efficiently utilize all replicas' bandwidth, leading to high throughput.

\subsection{Implementation Notes}\label{s:implementation-notes}


\begin{table}[t]
    \vskip 4pt 
    \centering
    \footnotesize{
    \caption{RTTs between regions (ms)}
    \label{tab:latencyrtt}
    \begin{tabular}{|c|c|c|c|c|}
           \hline RTT & us-east1 & us-east2 & us-west1 & us-west2 \\ \hline
           us-east1 & -  & 20 & 72 & 67 \\
           us-east2 & 21 & - & 57 & 55\\
          us-west1 & 69 & 53 & - & 23\\
          us-west2 & 66 & 49 & 24 & - \\ \hline
    \end{tabular}
    }
\end{table}
\par \textbf{Leader message sharing.}
For liveness, the leader needs to run both a fast-track lane and a replica lane. Instead of sending separate messages for each lane, a single message can be used with a boolean flag indicating the lane it belongs to.

\par \textbf{All-to-all communication.}
The recovery path utilizes linear style communication, where replicas forward their votes to the proposer instead of broadcasting them to all replicas. This reduces communication complexity at the trade-off of increasing latency. To reduce recovery path latency further, Ambulance supports an all-to-all voting strategy to shave off extra message delays.

\par \textbf{Signature aggregation.}
\sys{} uses threshold signatures, but only for the lane election phase. We can also adopt threshold signatures for protocol votes to reduce the communication complexity of \sys{} by a factor of $n$. In practice, this is only necessary for large values of $n$.

%% file: Evaluation/eval-cr.tex
\section{Evaluation}

\begin{table}[t!]
    \vskip 4pt 
    \centering
    \caption{Default timeout values used in production systems}
    \label{tab:timeouts}
    \footnotesize{
    \begin{tabular}{|c|c|}
    \hline
           System & Timeout  \\
           \hline
           CockroachDB~\cite{cockroach-db} & 2 s \\
           Microsoft CCF~\cite{ccf-timeout} & 5 s \\
           Neo4j~\cite{neo4j} & 7 s \\
           TiKV~\cite{tikv-timeout} & 10 s \\
           HotShot~\cite{priv-com-espresso} & 12 s  \\  \hline 
    \end{tabular}
    }
    \vskip 2pt
\end{table}
Our evaluation seeks to answer three questions:
\vspace{-1pt}
\begin{enumerate}
    \item \textbf{Performance}: How well does \sys{} perform in the absence of slowdowns? (\S\ref{s:eval-common-case})
    \item \textbf{Slowdown Robustness}: How well does \sys{} tolerate slowdowns? (\S\ref{s:eval-robustness})         
    \item \textbf{Real-world Performance}: How well does \sys{} perform in production under realistic conditions? (\S\ref{s:eval-production})
\end{enumerate}
\vspace{-1pt}

\begin{figure*}[!th]
\centering
\begin{minipage}{.6\textwidth}
  \centering
  \includegraphics[width=1\linewidth]{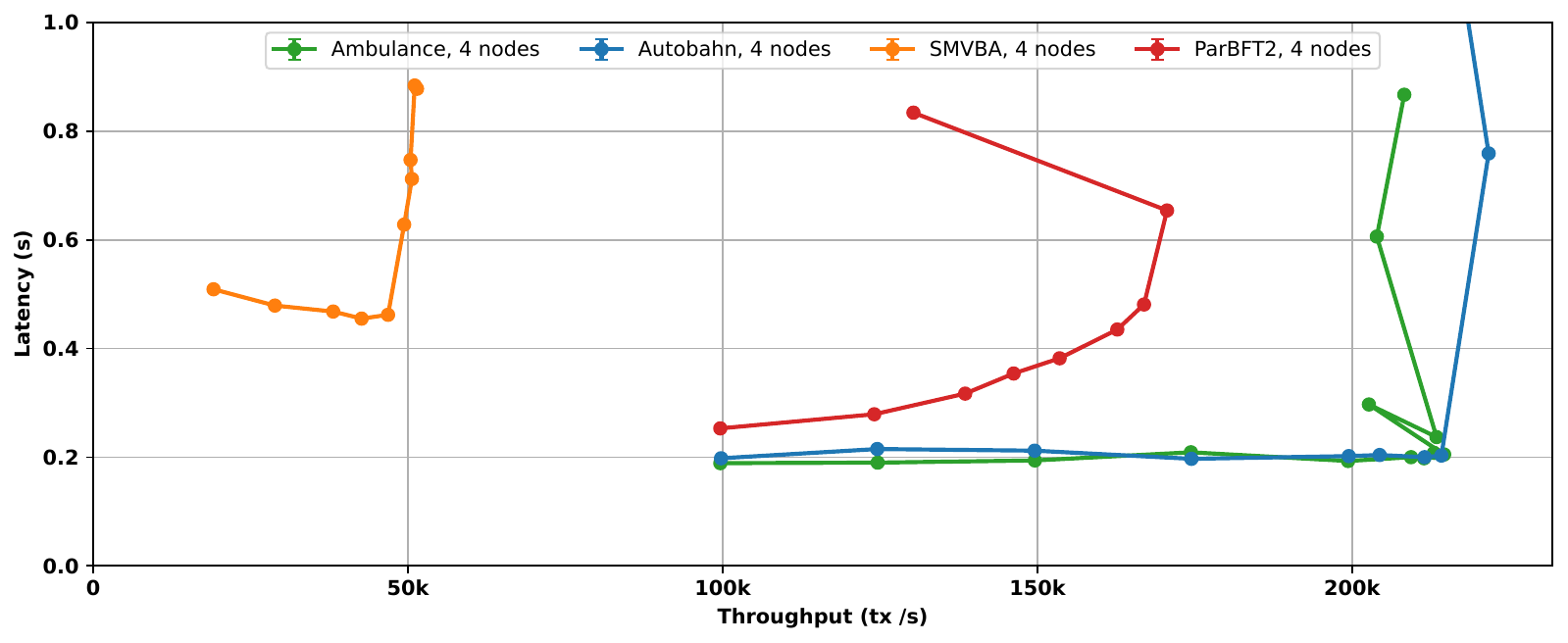}
  \vskip -2pt
  \caption{Throughput and Latency under increasing load.}
  \label{fig:performance}
\end{minipage}
\begin{minipage}{.39\textwidth}
  \centering
  \includegraphics[width=1\linewidth, height=0.62\linewidth]{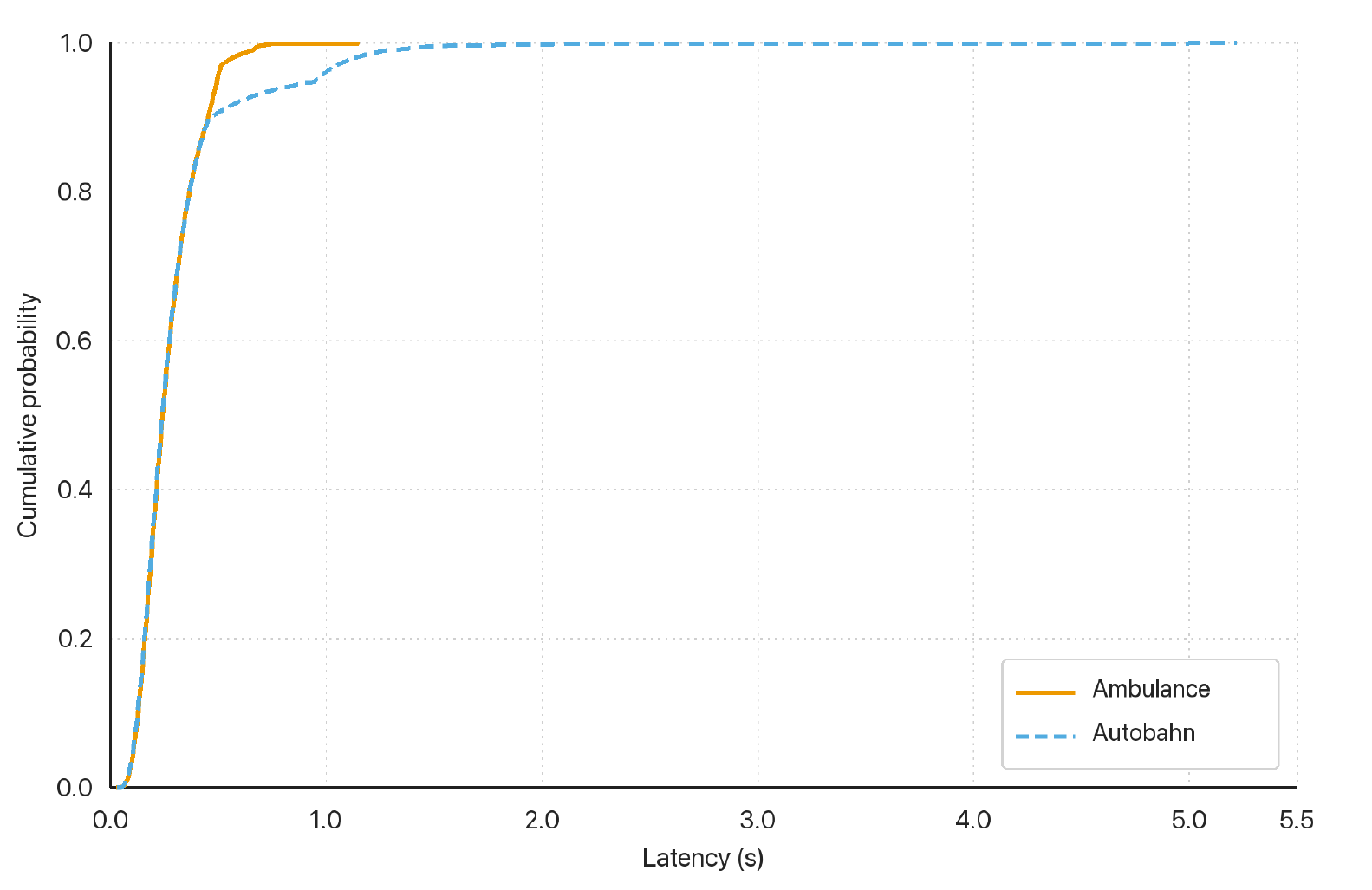}
  \vskip -0pt
  \caption{Production Latency CDF}
  \label{fig:cdf}
\end{minipage}
\vskip -10pt
\end{figure*}

We implemented a prototype of \sys{}\footnote{https://github.com/neilgiri/ambulance-artifact} in Rust, using the same codebase as the open source implementation of Autobahn~\cite{autobahn}. We implement all the notes in \S\ref{s:implementation-notes} except for signature aggregation. 
We use Tokio's TCP~\cite{tokio} for networking, and RocksDB~\cite{rocksdb} for persistent storage of payload data to disk. Finally, we use ed25519-dalek~\cite{ed25519} signatures for authentication. 

\par \textbf{Baseline systems.}
We compare \sys{} against several state-of-the-art BFT protocols: 1) Autobahn~\cite{autobahn}, a leading timeout-based BFT protocol, 
2) ParBFT2~\cite{parbft}, a state-of-the-art hedging-based protocol, as well as 3) SMVBA (Speeding Dumbo)~\cite{speeding-dumbo}, a state-of-the-art pessimistic (asynchronous) protocol. ParBFT2 uses HotStuff~\cite{yin2019hotstuff} as its optimistic path, and after a hedging delay, it launches its pessimistic path, which runs SMVBA combined with an asynchronous binary agreement protocol~\cite{abraham2022efficient}. We use the available open-source codebase each time\footnote{We validated our Autobahn and ParBFT2 results with respective authors}. 
\sys{} and Autobahn share the same data layer design, where all replicas disseminate batches through structured data lanes, enabling parallelized data dissemination. In contrast, ParBFT2 and SMVBA adopt the Batched HotStuff data layer~\cite{danezis2022narwhal}, where replicas optimistically stream batches and consensus leaders propose hashes of any batches they have received.

\sys{} and Autobahn use a PBFT-style multi-shot consensus with a totally ordered log, allowing multiple consensus instances to be pipelined in parallel. This eliminates sequential delays and reduces end-to-end latency. ParBFT2 and SMVBA instead follow the HotStuff~\cite{yin2019hotstuff} chaining approach, which introduces inherent sequential latency. In ParBFT2’s optimistic path, a new consensus instance begins only after the leader collects a quorum of votes for the previous one. In its pessimistic path-and in SMVBA—there is no pipelining at all; consensus instances run strictly sequentially. As a result, requests must wait for the next consensus instance to start, adding further end-to-end latency.

\par \textbf{Evaluation Setup.}
We evaluate all systems on AWS EC2. For sections \S\ref{s:eval-common-case} and \ref{s:eval-robustness}, we evenly distributed nodes in \texttt{us-west1}, \texttt{us-west2}, \texttt{us-east1} and \texttt{us-east2}. We summarize RTTs between regions in Table~\ref{tab:latencyrtt}.
We use machine type \texttt{m6a.4xlarge}~\cite{machineawsmsix} (16 vCPUs and 64GB RAM) with 30GB gp3 disk (SSD) and 12.5GB/s network bandwidth.
Clients are co-located in the same region as the replica, and issue a constant stream of no-op transactions (tx), consisting of 512 random bytes~\cite{danezis2022narwhal, Bullshark}; no-op transactions allow us to stress the consensus system, without risking being bottlenecked on execution. 
We set a batch size of 500KB (1000 transactions) for all baselines, but allow consensus proposals to include/reference more than one batch if available; this \textit{mini-batching} design~\cite{codehotstuff,codebullshark} allows replicas to organically reach larger effective batch sizes with reduced latency trade-off.
For the systems that use leaders, they are elected using a deterministic round-robin scheme.

\subsection{Common Case Performance}\label{s:eval-common-case}
Figure~\ref{fig:performance} presents performance in a fault-free, synchronous setting with $n=4$ nodes. Each experiment runs for 60 seconds under increasing client load. \sys{} reaches a peak throughput of 214k tx/s, matching Autobahn’s peak throughput of 214k tx/s. This high throughput stems from motorization as \sys{} adopts Autobahn’s data layer. \sys makes efficient use of network bandwidth, which dominates overall cost, as consensus messages are small compared to large batched proposals. Although \sys{} introduces multiple consensus lanes, the metadata overhead is negligible in practice, and both \sys{} and Autobahn ultimately bottleneck on data layer work (serializing/deserializing).

\sys{} achieves a latency of 205 ms compared to Autobahn’s 203 ms. In the absence of failures and synchrony, both systems essentially run PBFT, leading to similar latency.



ParBFT2 reaches a peak throughput of 167k tx/s, substantially higher than SMVBA but 1.3x lower than \sys. This gap is primarily due to its weaker data layer, which incurs substantial data-synchronization overhead at high load~\cite{autobahn, danezis2022narwhal}, limiting throughput even when consensus is not the bottleneck. In the good case, ParBFT2's consensus latency is five message delays, compared to three for \sys{}, contributing to its 382 ms latency (1.9x higher than \sys).

SMVBA achieves only 50.6k tx/s peak throughput and exhibits the highest latency of 462 ms (2.3x higher than \sys). It shares the same data layer as ParBFT2, but its low throughput follows directly from its sequential consensus design: consensus instances run strictly one after another, so throughput is constrained by consensus latency. In the good case, SMVBA's consensus latency is six message delays, more than ParBFT2 (five) and \sys{} (three). Together with the additional inclusion latency introduced by this sequential design, these extra message delays yield higher end-to-end latency than \sys.

Overall, \sys{} matches the common-case throughput and latency of state-of-the-art timeout-based BFT (Autobahn) while outperforming pessimistic and hedging-based systems.

\begin{figure*}[!t]
\hspace*{-0.5cm}   
\centering
\begin{minipage}{0.48\textwidth}
  \centering
  \includegraphics[width=1\linewidth]{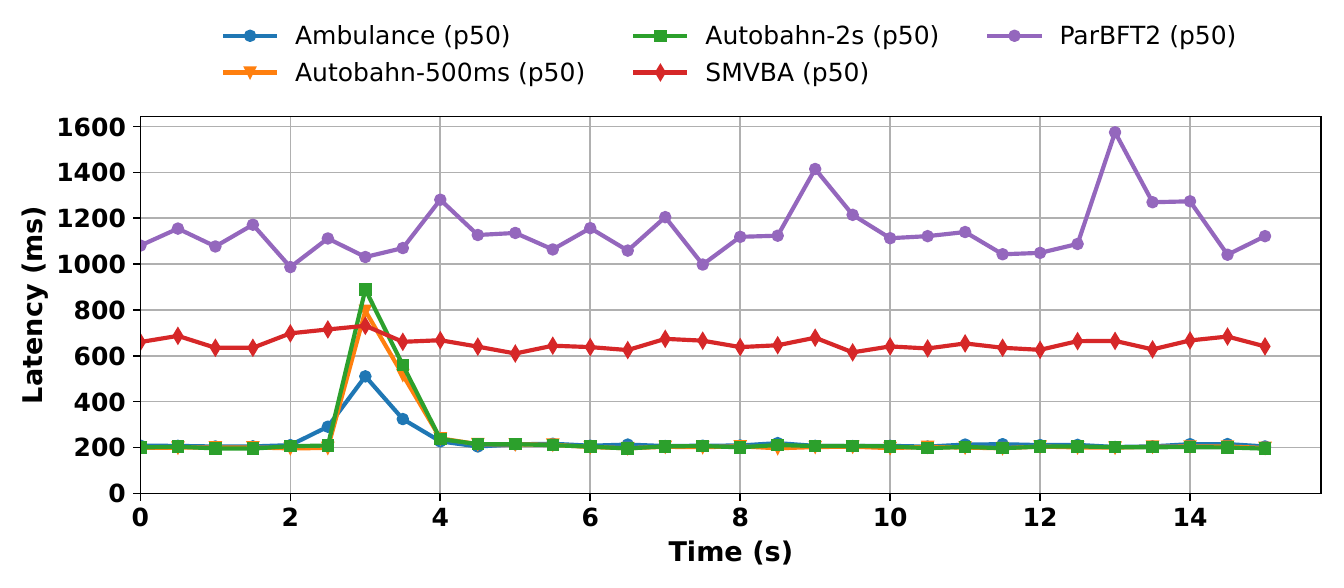}
    \vskip -3pt
  \caption{1s slowdown}
  \label{fig:slowdown-1s}
\end{minipage}
\hspace{4mm}
\begin{minipage}{0.48\textwidth}
   \centering
    \includegraphics[width=1\linewidth]{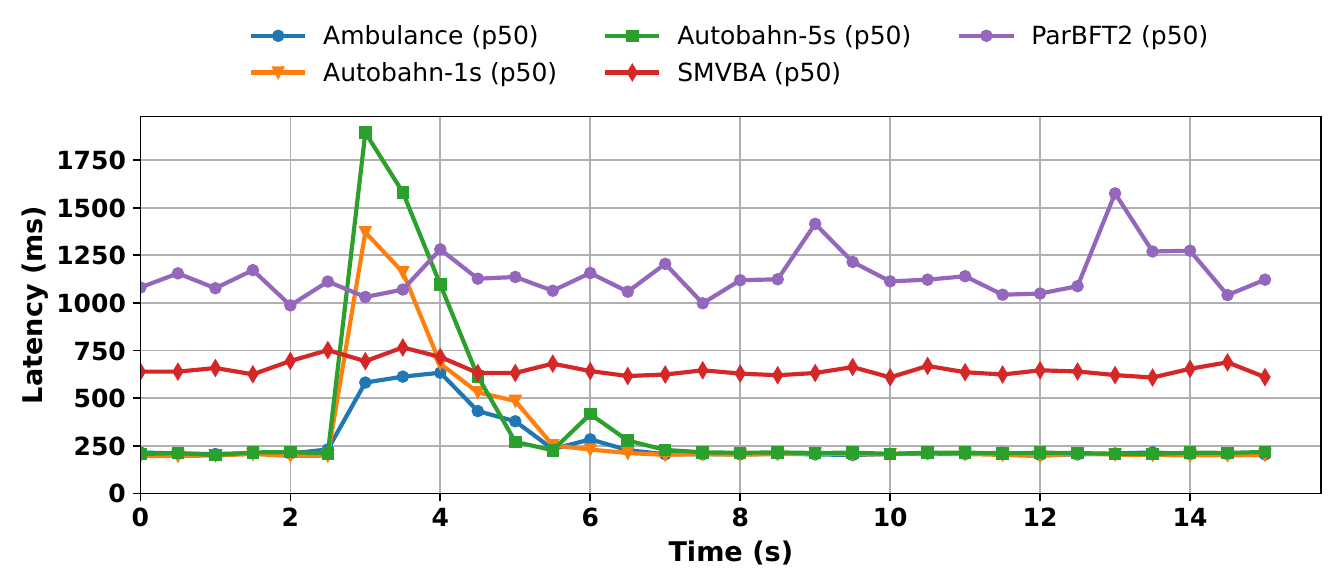}
   \vskip -3pt
   \caption{2s slowdown}
   \label{fig:slowdown-2s}
\end{minipage}

\vskip -10pt
\end{figure*}

\begin{figure*}[!th]
\centering
\begin{minipage}{0.33\textwidth}
  \centering
    \includegraphics[width=1\linewidth]{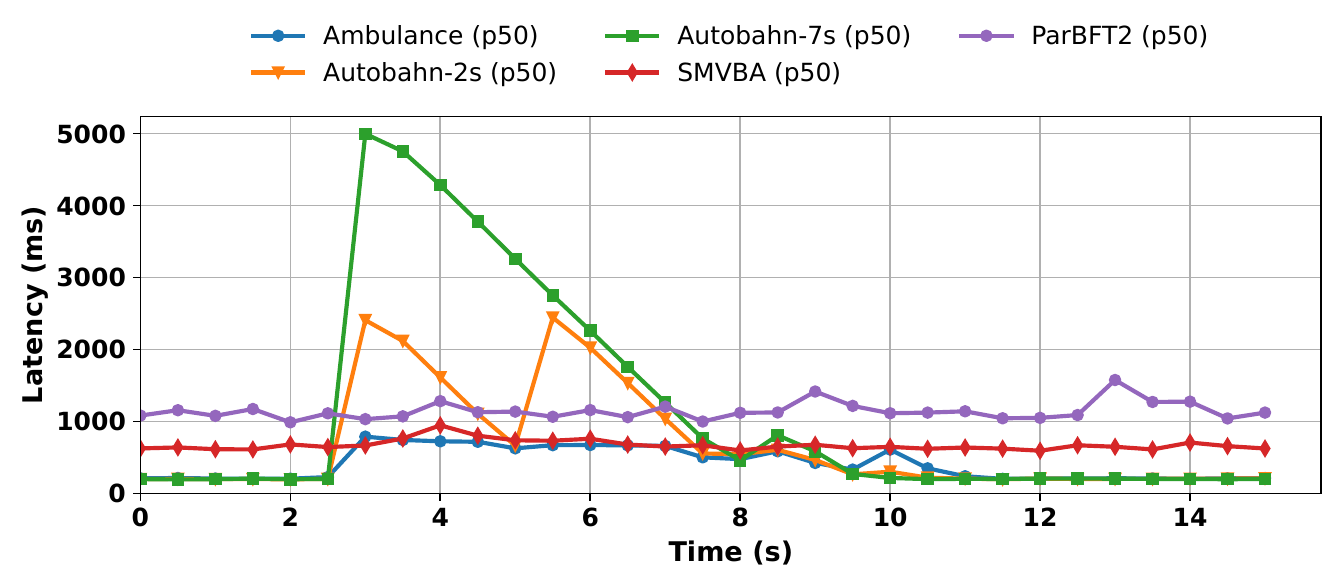}
  \vskip -3pt
  \caption{5s slowdown}
  \label{fig:slowdown-5s}
\end{minipage}
\begin{minipage}{.33\textwidth}
  \centering
  \includegraphics[width=1\linewidth]{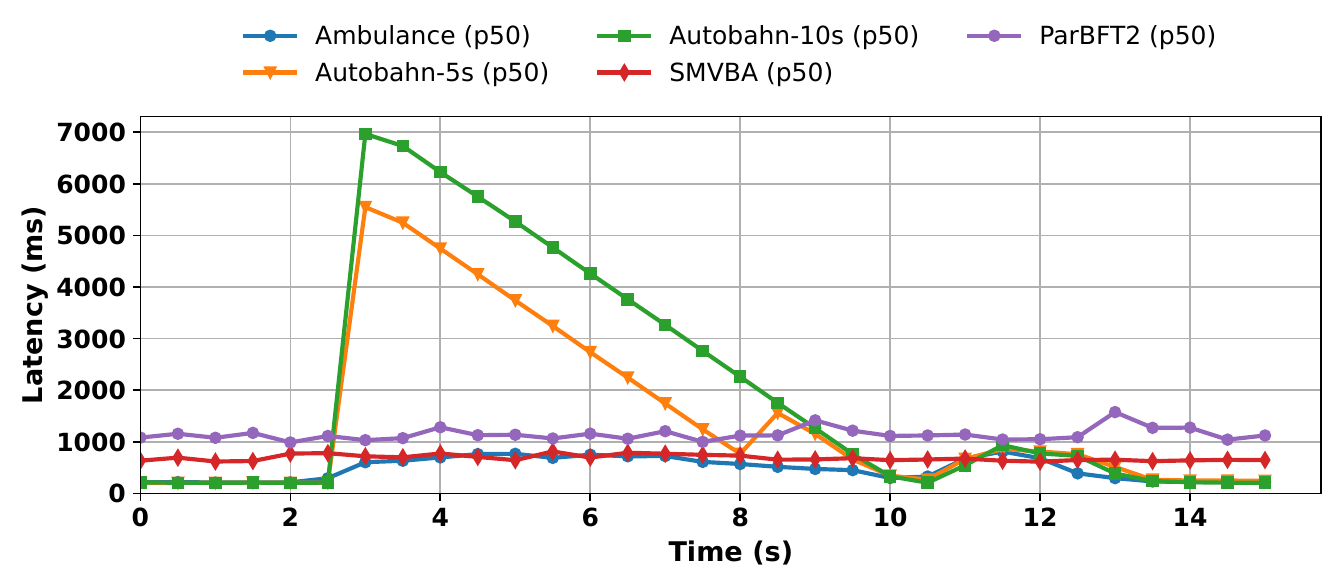}
  \vskip -2pt
  \caption{7s slowdown}
  \label{fig:slowdown-7s}
\end{minipage}
\begin{minipage}{.33\textwidth}
  \centering
  \includegraphics[width=1\linewidth]{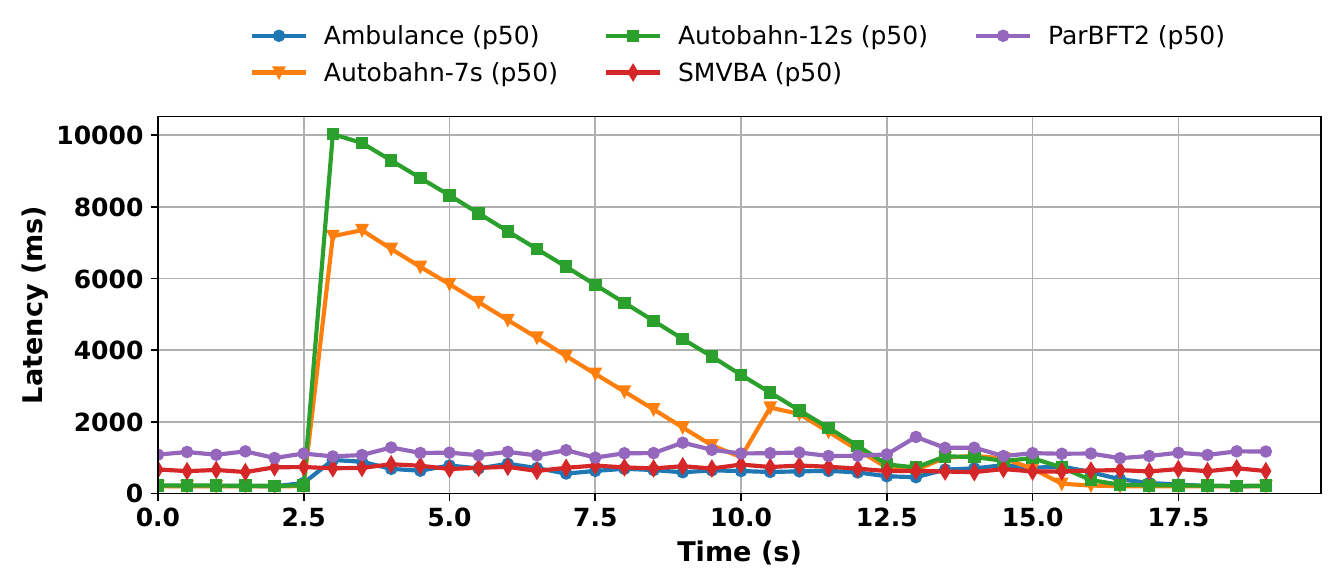}
  \vskip -0pt
  \caption{10s slowdown}
  \label{fig:slowdown-10s}
\end{minipage}
\vskip -10pt
\end{figure*}

\subsection{Slowdown Robustness}\label{s:eval-robustness}
The previous section demonstrated that \sys{} performs well under ideal conditions.
We now examine \sys{}’s tolerance to slowdowns of varying severity in an $n=4$ deployment. At time $t = 3$ seconds, we inject a transient slowdown whose duration we vary across experiments. The slowdown is introduced by pausing all processing (i.e., sleeping) on a single replica.
For ParBFT2, we run only the pessimistic path and do not inject slowdowns, due to bugs in the optimistic–pessimistic path switching logic.\footnote{We are working with the authors of~\cite{parbft} to resolve these issues.} This configuration still captures slowdown latency but does not reflect ParBFT2’s common-case latency, which should closely match Figure~\ref{fig:performance} (382 ms). 

We plot latency over time for slowdown durations of 1, 2, 5, 7, and 10 seconds, where each point is the median latency over a 500 ms window across 5 runs. For Autobahn, we use the timeout values from Table~\ref{tab:timeouts}, which are taken from production deployments of major reliable datastores. 

In the slowdown experiments, we run Autobahn with two different leader-election timeout values. One timeout value is set above the injected slowdown (timeouts of 2, 5, 7, 10, and 12 seconds for slowdowns of 1, 2, 5, 7, and 10 seconds, respectively) so that the plots capture the full effect of the slowdown. The other timeout value is set below the slowdown duration, so that the timeout detects the slowdown. Using a more aggressive timeout for a given slowdown duration caps the peak latency at the timeout value, making it closely match the peak latency of a smaller slowdown. For ParBFT2, we use an aggressive 75 ms hedging delay to minimize slowdown latency, consistent with the maximum RTT in Table~\ref{tab:latencyrtt} (\cite{parbft} specifies a hedging delay of $2\Delta$, equal to the maximum RTT).

\par \textbf{1 second slowdown (Figure~\ref{fig:slowdown-1s}).} All systems are evaluated under a conservative load (\cite{copilot}) to minimize queuing effects (\sys{} and Autobahn at 100k, and ParBFT2 and SMVBA at 10k).
\sys{} has the lowest peak latency at 510.4 ms. This is 1.6x better than Autobahn with a 500 ms timeout (796.1 ms), 1.7x better latency than Autobahn with a 2~s timeout (889.5 ms), 1.4x better than SMVBA (731 ms), and 3.1x better than ParBFT2 (1.58 s). \sys{} promptly detects the slowdown and switches to the recovery path, which takes 7 message delays for the first view. Autobahn, configured with a 2~s timeout (from Table \ref{tab:timeouts}),
does not detect the slowdown, causing latency to rise proportionally to the slowdown duration as replicas block on the slow leader.
Autobahn with a 500 ms timeout detects the slowdown; however, it must elect a new leader and pay the consensus latency cost, so the overall latency ends up being similar to Autobahn with a 2~s timeout.
SMVBA’s consensus latency is similar to \sys{}’s, but its lack of pipelining increases end-to-end inclusion latency. ParBFT2 exhibits the highest latency due to the hedging delay and the requirement to run both SMVBA and an asynchronous binary agreement protocol for its pessimistic path.

\par \textbf{2 second slowdown (Figure~\ref{fig:slowdown-2s}).}
With a 2 second slowdown, \sys{}’s peak latency increases mildly to 633.4 ms, but is still 2.2x better than Autobahn with a 1~s timeout (1.37 s), 3.0x better than Autobahn with a 5~s timeout (1.89 s), 1.2x better than SMVBA (766 ms), and 2.5x better than ParBFT2 (1.58s). 
Longer slowdowns increase the probability that multiple recovery views are needed, raising \sys{} and SMVBA's latency. Autobahn with a 5~s timeout
again fails to detect the slowdown, resulting in latency degradation proportional to the slowdown duration. Autobahn with a 1~s timeout detects the slowdown after 1 second, but must also pay the consensus latency afterwards, so the latency is still higher compared to \sys.
ParBFT2 behaves similarly to the 1 s case.
Figures~\ref{fig:slowdown-5s}–\ref{fig:slowdown-10s}.

\par \textbf{5-10 second slowdowns (Figures~\ref{fig:slowdown-5s}-\ref{fig:slowdown-10s})}
The same trend holds. \sys{} reaches peak latencies of 787.7 ms, 810 ms, and 932 ms for 5, 7, and 10 second slowdowns, respectively, increasing only modestly due to additional failed views. These are still 3.1-6.3x, 6.9-8.6x, and 7.9-10.8x better than Autobahn for 5, 7, and 10 second slowdowns, respectively. SMVBA closely tracks \sys{} at these longer slowdowns (with peak latencies of 812 ms and 945 ms), while ParBFT2 remains at 1.58 s (1.7x worse than \sys).

Across all slowdown durations, \sys{} consistently delivers lower peak latency than Autobahn—1.6x–3.0x better for short (1–2 second) slowdowns and up to 10.8x better for more severe (5–10 second) slowdowns. \sys{} also closely matches SMVBA, the state-of-the-art pessimistic protocol, while achieving 1.7x–3.1x lower latency than ParBFT2, the state-of-the-art hedging protocol.


\subsection{Production Performance}\label{s:eval-production}
We evaluate the performance of \sys{} as deployed in Sei~\cite{marsh2025seigiga}, a production-grade distributed ledger system, and compare it to their current implementation of Autobahn. The deployment consists of $n = 40$ replicas distributed across 20 AWS EC2 regions in North America, South America, Asia, and Europe. Each replica runs on an \texttt{m6i.12xlarge}~\cite{unrelatedmachine} instance (48 vCPUs, 192 GB RAM, 18.75 GB/s network bandwidth).

Figure~\ref{fig:cdf} reports latency over a 24-hour production run. The system experiences slowdowns approximately once every thousand slots, consistent with rates observed in deployed distributed ledger systems~\cite{mystenprivate, priv-com-espresso, marsh2025seigiga}. Both \sys{} and the production Autobahn implementation operate at 180k load (with a peak throughput of 220k) on the same workload; Autobahn uses 2-second timeouts.
\sys{} matches the median latency of Autobahn (244ms vs 242ms), while significantly improving tail performance: its p99 latency of 660 ms represents a 1.94x reduction compared to Autobahn’s 1.28 s. Because the slowdown events dominate the p99, these results show that \sys{} maintains Autobahn’s common-case latency in production while substantially reducing latency during slowdowns.

%% file: Related-Work/related_work.tex
\section{Related Work}

\sys{} is inspired by several prior works. \sys{}'s leader racing path draws from traditional BFT protocols, in particular, PBFT~\cite{castro1999pbft}. The recovery path draws from the VABA (MVBA) family of asynchronous protocols, in particular, 2PAC~\cite{2pac}. \sys{} directly adopts Autobahn's~\cite{autobahn} data layer and pipelining strategy to achieve high throughput and reduce end-to-end latency.

\par \textbf{Slowdown tolerance and failure detectors}
The most common failure detectors~\cite{weakestfd,unreliablefd,timelinessfd} rely on heartbeats or timeouts. These mechanisms are slow to react to failures, and are fundamentally not cooperative nor productive. Other failure detection mechanisms~\cite{notime-async, watchdogs, failure-lorch, falcon-spy} are less generic and inspect the actual processes themselves. These mechanisms can detect more types of failures, including gray failures, but do not easily extend to Byzantine environments.
PeerReview~\cite{peerreview} does detect provably faulty behavior but can neither detect nor correct slow behavior. 
Copilot~\cite{copilot} is the first SMR protocol to introduce slowdown tolerance, limited to the CFT setting. It defines and satisfies the $s$-slowdown-tolerance property, which guarantees performance similar to the no-slowdown case even with $s$ slow replicas. We instead assess slowdown robustness by measuring each protocol’s latency under slowdown, since even Copilot can suffer performance degradation in certain slowdown scenarios~\cite{slowdown-ryan}. 
\par \textbf{Partially synchronous BFT} These protocols~\cite{kotla10zyzzyva, castro1999pbft, simplex, yin2019hotstuff, giridharan2023beegees, gueta2019sbft, giridharan2021no, kauri, stathakopoulou2019mir, gupta2021rcc, spiegelman2022bullshark, sailfish, shoalplusplus, babel2023mysticeti} are optimized for the failure-free case, minimizing latency when there are no slowdowns. Because they rely on timeouts to detect failures, a slow leader causes latency to spike by an amount proportional to the leader timeout or the duration of the slowdown.
\par \textbf{Robust BFT}
Aardvark~\cite{aardvark} monitors leader performance and invokes view changes if the leader does not meet performance thresholds. The system offers only partial protection against slowdowns~\cite{copilot} as it cannot easily detect transient slowdowns. 
Autobahn~\cite{autobahn} is robust to blips, a subset of slowdowns that includes network events such as partial partitions or replica failures that trigger timeouts. It ensures that latency returns to baseline immediately after a blip but allows latency to spike during the blip. 
\par \textbf{Asynchronous BFT} protocols~\cite{mahi-mahi, danezis2022narwhal, vaba, honeybadger} are pessimistic: they assume the leader always loses the race against the clock. There are two main families of asynchronous BFT protocols: the HoneyBadger (BKR) family~\cite{honeybadger, liu2023flexible, dispersed-ledger, aleabft} and the VABA family~\cite{vaba, gao2022dumbo, spiegelman_ace, speeding-dumbo, 2pac}. HoneyBadger-style protocols typically run $n$ instances of reliable broadcast~\cite{bracha1985asynchronous} and $n$ instances of binary agreement to decide which broadcasts completed, while VABA-style protocols run $n$ instances of a leader-based BFT protocol and use a common coin to randomly select one instance to commit. Asynchronous protocols have lower latency than traditional BFT during slowdowns, since they do not require a failure detection mechanism such as timeouts. However, they have much higher common-case latency than traditional BFT protocols, limiting their adoption in practice.

\par \textbf{Hybrid BFT} protocols~\cite{cachin06opt, jolteon-ditto, dumbo-transformer, Bullshark, parbft, abraxas, ipotane} aim to combine the advantages of traditional and asynchronous BFT. Ditto~\cite{jolteon-ditto} and BDT~\cite{dumbo-transformer} run a traditional BFT protocol first and fall back to an asynchronous one when a timeout fires. Relying on timeouts causes them to suffer high latency during slowdowns, particularly when switching protocols. ParBFT~\cite{parbft}, Abraxas~\cite{abraxas}, and Ipotane~\cite{ipotane} instead run one or more asynchronous protocol instances to detect failures of the traditional BFT protocol. Unfortunately, commitment during a slowdown requires still requires completing multiple asynchronous instances.
Hedging~\cite{tennage2023quepaxa, parbft} detects failures by having the leader race against other replicas, but giving the leader a time-based head start. This mechanism is cooperative, but not productive; it stalls recovery until the hedging delay is over.

\par \textbf{Leaderless SMR} protocols~\cite{epaxos, bpaxos, atlas} distribute the ordering task across multiple replicas instead of relying on a single leader. Because no single leader orders all requests, replicas commit dependency graphs that induce an order among instances. This design avoids stalling on a single slow leader, but slowdowns still harm performance: replicas can depend on commands proposed by slow replicas and must wait for those commands to commit before executing. In the worst case, they must invoke an expensive recovery procedure for those commands.

%% file: Conclusion/conclusion.tex
\section{Conclusion}
This paper introduces \sys, a novel BFT protocol, that uses a novel \textit{protocol-rigged racing} mechanism to detect slowdowns. This allows \sys{} to offer the best of both worlds: the good-interval performance of timeout-based systems and slowdown performance comparable to asynchronous protocols.

%% file: Protocol/retry_protocol_appendix.tex
\section{Retry Protocol}\label{s:retry}
\par \textbf{Recovery Input Selection.} 
To determine whether a value may have been committed in prior views, replicas broadcast their local state to others.

\par \protocol{\textbf{1: P} $\rightarrow$ \textbf{R}: Replica $P$ broadcasts \viewchange.}
When $P$ enters view $v$, it broadcasts
its local understanding of whether the previously elected lane committed (it sets $\certdata=qcs[E]$ or $\noelectcert{_P}$ if it did not receive an \elected \preptwocertname): if $\certdata=\noelectcert{_P}$, no operation committed in view $v-1$ according to $P$ (note that, unlike traditional BFT protocols, this says nothing about whether a value committed in a prior view). Otherwise, it forwards  $\certdata=qcs[E]$ to all other replicas as part of a $\viewchange \coloneqq \langle v, \certdata \rangle$ message.

\par \protocol{\textbf{2: P}:  $P$ receives a quorum of valid \viewchange messages.} 
Replica $P$ decides what input to recover based on the $n-f$ valid \viewchange messages it has received. From these, $P$ distinguishes between two cases: 1) a value could have been committed in view $v-1$ 2) a value definitely did not commit in view $v-1$, but may have committed in a prior view.
\par \protocol{\textbf{Case 1}: $P$ receives a \preptwocertname in view $v-1$.}
Upon receiving a \preptwocertname (\preptwocert), $P$ checks that is well-formed and from the elected lane in view $v-1$. Receiving an \preptwocertname signals that a correct replica \textit{may have committed} on the recovery path for view $v-1$.  This follows from quorum intersection: if an \elected \confirmcertname exists, $P$ is guaranteed to receive at least one corresponding \preptwocertname. To preserve safety, $P$ must therefore set its recovery input to the \elected \preptwocertname:
$\textit{r-input} = \preptwocert$.

\par \protocol{\textbf{Case 2}: $P$ receives $n-f$ \noelectcert{}.}
Receiving $n-f$ \noelectcert{} messages definitively indicates that no correct replica committed on the recovery path in view $v-1$. By quorum intersection, if an \elected \confirmcertname had existed, $P$ would necessarily have received an \preptwocertname.

However, $P$ still cannot safely recover an arbitrary value. Receiving $n-f$ \noelectcert{} says nothing about what was decided in \textit{prior views}. If a value was previously committed, $P$ must recover the value and set it as its recovery input. We distinguish between two cases. A correct replica may have committed a value 1) on the racing path 2) on the recovery path in any earlier view $v'<v-1$.

If a value $x$ committed on the racing path, all replicas would have received a \lockcertname in view $0$, and thus set this as their recovery input. From then onwards, all \confirmcertname exchanged as part of higher views will always be for value $x$. 
All \confirmcertname received as part of view $v-1$ will thus necessarily be for value $x$. $P$ can select any as its recovery input.  

If a value $x$ instead committed in view $v'<v-1$, all replicas would have received a \recoveryprepare for $x$ in view $v'+1$, and would have thus set their recovery input to $x$. From then onwards, all \confirmcertname exchanged as part of higher views will always be for value $x$. 
In both cases, $P$ can thus set its recovery input to any \confirmcertname from $v-1$: $\textit{r-input} = \confirmcert$, and aggregates the $n-f$ signatures on \noelectcert{} into a \noelectcertname, \necert.


After selecting its recovery input, $P$ enters the \recoveryprepphase{} phase, using the \noelectcertname{} in place of the \nocommitcert{} (if it exists). From this point forward, $P$ executes the persistence and lane-election phases exactly as in view $0$, with no further differences.


%% file: Proofs/new-proofs.tex
\section{Proofs}\label{s:proofs}

\subsection{Safety}

For safety, correct replicas enforce the following validity rules for messages for a given slot.

\paragraph{Protocol validity rules.}
All rules below are for a fixed consensus instance/slot. Correct replicas only
sign or vote for messages that are valid according to the rules below.

\begin{enumerate}
    \item \textbf{Fast-prepare uniqueness.}
    A correct replica sends at most one \fastprepvote message.

    \item \textbf{No-leader-proposal exclusion.}
    A correct replica never both signs \noleaderprop{} and sends a
    \fastprepvote message.

    \item \textbf{Fast-commit validity.}
    A correct replica sends \fastcommit for digest $d$ only after verifying
    and storing a valid \lockcertname{} for the same digest $d$.

    \item \textbf{No-leader-lock exclusion.}
    A correct replica never both sends \fastcommit and signs \noleaderlock.
    This exclusion is permanent in both directions.

    \item \textbf{Slow-prepare uniqueness.}
    For each proposer $i$ in view $0$, a correct replica sends at most one
    \slowprepvote message.

    \item \textbf{Recovery-prepare vote uniqueness.}
    For each proposer $i$ and view $v$, a correct replica sends at most one
    \recoveryprepvote message.

    \item \textbf{Recovery-prepare vote validity.}
    A correct replica sends a \recoveryprepvote for digest $d$ only after
    verifying a valid \recoveryprepare message for the same proposer, view,
    and digest $d$.

    \item \textbf{View-$0$ recovery-prepare validity.}
    In view $0$, a \recoveryprepare message from proposer $i$ for digest $d$
    is valid only if it contains one of the following:
    \begin{enumerate}
        \item a valid \lockcertname{} for digest $d$; or
        \item a valid \slowcertname{} from proposer $i$ for digest $d$ together
        with a valid \nocommitcert{}.
    \end{enumerate}
    In particular, the digest of the \recoveryprepare must match the digest of
    the supporting certificate.

    \item \textbf{Later-view recovery-prepare validity.}
    For view $v>0$, a \recoveryprepare message for digest $d$ is valid only if
    it adopts one of the following certificates from view $v-1$:
    \begin{enumerate}
        \item an \elected \preptwocertname{} from view $v-1$ with digest $d$; or
        \item a non-\elected \confirmcertname{} from view $v-1$ with digest $d$
        together with a valid \noelectcertname{} for view $v$.
    \end{enumerate}
    Thus, the digest proposed in view $v$ must equal the digest of the
    certificate adopted from view $v-1$.

    \item \textbf{Recovery-confirm validity.}
    A correct replica sends a \recoveryconfirmvote for digest $d$ only after
    verifying a valid \recoveryconfirm message containing a valid
    \preptwocertname{} for the same proposer, view, and digest $d$.

    \item \textbf{Recovery-confirm vote uniqueness.}
    For each proposer $i$ and view $v$, a correct replica sends at most one
    \recoveryconfirmvote message.

    \item \textbf{Election uniqueness.}
    In every view $v$, all correct replicas agree on at most one elected
    proposer. A certificate is \elected in view $v$ only if it originates from
    that proposer.

    \item \textbf{No-election exclusion.}
    A correct replica never both sends a \recoveryconfirmvote for an
    \elected \preptwocertname{} in view $v$ and signs \noelectcert for view
    $v+1$. This exclusion is permanent in both directions.

    \item \textbf{Commit validity.}
    A correct replica commits on the fast path only after receiving a valid
    \leadercommitname{}. A correct replica commits on the recovery path only
    after receiving a valid \elected \confirmcertname{}.

    \item \textbf{Embedded validity.}
    Whenever a correct replica verifies a message containing an embedded certificate, the embedded certificate must be valid and must match the same consensus instance, proposer where applicable, view, and digest claimed by the outer message.
\end{enumerate}

\begin{lemma}\label{lem:leader-lock}
For any two \lockcertname, $P$ and $P'$, $P.dig=P'.dig$
\end{lemma}

\begin{proof}
Suppose for the sake of contradiction, $P.dig\neq P'.dig$. $P$ and $P'$ consist of $n-f$ \fastprepvote messages. Two quorums of size $n-f$ intersect in at least one correct replica, which means at least one correct replica sent conflicting \fastprepvote messages with different digests, a contradiction since correct replicas only send one \fastprepvote message.
\end{proof}

\begin{lemma}\label{lem:leader-commit}
For any two \leadercommitname, $C$ and $C'$, $C.dig=C'.dig$
\end{lemma}

\begin{proof}
Suppose for the sake of contradiction, $C.dig\neq C'.dig$. $C$ and $C'$ consist of $n-f$ \fastcommit messages. Two quorums of size $n-f$ intersect in at least one correct replica, which means at least one correct replica sent conflicting \fastcommit messages for different digests. A correct replica only sends a \fastcommit message for digest, $dig$, if it receives a \lockcertname also for $dig$, meaning it received two conflicting \lockcertname, a contradiction to lemma \ref{lem:leader-lock}.
\end{proof}

\begin{lemma}\label{lem:leader-lock-commit}
For any \lockcertname, $L$, and any \leadercommitname, $C$, $L.dig=C.dig$.
\end{lemma}

\begin{proof}
Suppose for the sake of contradiction, $L.dig\neq C.dig$. By lemma \ref{lem:leader-lock} any two \lockcertname must have the same digest, $L.dig$. $C$ consists of $n-f$ \fastcommit messages, of which at least $n-2f$ are from correct replicas. This means that $n-2f$ correct replicas sent a \fastcommit message for a digest that does not correspond to a valid \lockcertname, a contradiction.  
\end{proof}

\begin{lemma}\label{lem:replica-non-equiv}
For any two \slowcertname, $P$ and $P'$, from proposer $i$ in view $0$, $P.dig=P'.dig$
\end{lemma}

\begin{proof}
Suppose for the sake of contradiction, $P.dig\neq P'.dig$. $P$ and $P'$ consist of $n-f$ \slowprepvote messages. Two quorums of size $n-f$ intersect in at least one correct replica, which means at least one correct replica sent conflicting \slowprepvote messages with different digests, a contradiction since correct replicas only send one \slowprepvote message for a given proposer in view $0$.
\end{proof}

\begin{lemma}\label{lem:no-lock-cert}
If a \nolockcert{} forms, then no \lockcertname{} can exist.    
\end{lemma}

\begin{proof}
Suppose for the sake of contradiction there exists both a \nolockcert and a \lockcertname. This means a quorum $Q_1$ of size at least $n-f$ signed \noleaderprop. And a quorum $Q_2$ of size at least $n-f$ sent a \fastprepvote message. $Q_1$ and $Q_2$ intersect in at least one correct replica, which signed \noleaderprop and sent a \fastprepvote message. A correct replica will not do both, a contradiction.
\end{proof}

We say that a \preptwocertname is either $n-f$ \recoveryprepvote messages, or a tuple consisting of $n-f$ \slowprepvote messages and a \nolockcert{} (upgraded).

\begin{lemma}\label{lem:prep2-non-equiv}
For any two \preptwocertname, $P$ and $P'$, in view $0$ from proposer $i$, $P.dig=P'.dig$    
\end{lemma}

\begin{proof}
Suppose for the sake of contradiction, $P.dig\neq P'.dig$. There are three cases: \one $P$ and $P'$ consist of $n-f$ \recoveryprepvote messages \two $P$ and $P'$ consist of $n-f$ \slowprepvote messages and a \nolockcert{} or \three WLOG $P$ consists of $n-f$ \recoveryprepvote messages and $P'$ consists of $n-f$ \slowprepvote messages and a \nolockcert{}.

\par \textbf{Case \one.} $P$ and $P'$ consist of $n-f$ \recoveryprepvote messages. Two quorums of size $n-f$ intersect in at least one correct replica, which means at least one correct replica sent conflicting \recoveryprepvote messages with different digests, a contradiction since correct replicas only send one \recoveryprepvote message for a given proposer in view $0$.

\par \textbf{Case \two.} By lemma \ref{lem:replica-non-equiv} any two \slowcertname must be for the same digest, a contradiction.

\par \textbf{Case \three.} By lemma \ref{lem:replica-non-equiv}, any \slowcertname from proposer $i$ must have digest = $P'.dig$, and by lemma \ref{lem:no-lock-cert}, there does not exist a \lockcertname. This means that at least $n-2f$ correct replicas sent a \recoveryprepvote for a \recoveryprepare that did not contain a supporting \slowcertname or a \lockcertname, a contradiction since a valid \recoveryprepare must contain either a \slowcertname or a \lockcertname.
\end{proof}

\begin{lemma}\label{lem:confirm-non-equiv}
For any two \confirmcertname, $C$ and $C'$, in view $0$ from proposer $i$, $C.dig=C'.dig$    
\end{lemma}

\begin{proof}
Suppose for the sake of contradiction, $C.dig\neq C'.dig$. $C$ and $C'$ consist of $n-f$ \recoveryconfirmvote messages. Two quorums of size $n-f$ intersect in at least one correct replica, which means at least one correct replica sent conflicting \recoveryconfirmvote messages for different digests. A correct replica only sends a \recoveryconfirmvote message for digest, $dig$, if it receives a \preptwocertname also for $dig$ (contained in a \recoveryconfirm message), meaning it received two conflicting \preptwocertname, a contradiction to lemma \ref{lem:prep2-non-equiv}.    
\end{proof}

\begin{lemma}\label{lem:prep2-confirm-equal}
For any \preptwocertname, $P$, and \confirmcertname, $C$ in view $0$ from proposer $i$, $P.dig=C.dig$.
\end{lemma}

\begin{proof}
Suppose for the sake of contradiction, $P.dig\neq C.dig$. By lemma \ref{lem:prep2-non-equiv} any \preptwocertname must have digest equal to $P.dig$. This means that $n-2f$ correct replicas sent a \recoveryconfirmvote message for a digest not corresponding to a valid \preptwocertname, a contradiction.
\end{proof}

\begin{lemma}\label{lem:prepare-non-equiv}
For any two \preptwocertname $P$ and $P'$, from proposer $i$ in view $v>0$, $P.dig=P'.dig$.
\end{lemma}

\begin{proof}
Suppose for the sake of contradiction, $P.dig\neq P'.dig$. $P$ and $P'$ consist of $n-f$ \recoveryprepvote messages. Two quorums of size $n-f$ intersect in at least one correct replica, which means at least one correct replica sent conflicting \recoveryprepvote messages with different digests, a contradiction since correct replicas only send one \recoveryprepvote message per proposer per view.
\end{proof}

\begin{lemma}\label{lem:prepare-confirm-equal}
For any \preptwocertname, $P$, and \confirmcertname, $C$ in view $v>0$ from proposer $i$, $P.dig=C.dig$.
\end{lemma}

\begin{proof}
Suppose for the sake of contradiction, $P.dig\neq C.dig$. By lemma \ref{lem:prepare-non-equiv} any \preptwocertname from proposer $i$ in view $v$ must have digest equal to $P.dig$. This means that $n-2f$ correct replicas sent a \recoveryconfirmvote message for a digest that does not correspond to a valid \preptwocertname, a contradiction.
\end{proof}

\begin{lemma}\label{lem:commit-no-commit}
If there exists a \leadercommitname, $C$, then a \nocommitcert{} cannot exist.
\end{lemma}

\begin{proof}
Suppose for the sake of contradiction, there exists both a \leadercommitname and a \nocommitcert{}. The existence of $C$ implies that at least $n-2f$ correct replicas sent a \fastcommit message. These correct replicas must have stored a \lockcertname locally (in \lockcert) before sending a \fastcommit message. However, the existence of a \nocommitcert{} implies that at least $n-f$ replicas signed the message \noleaderlock. A quorum of $n-2f$ correct replicas intersects any quorum of $n-f$ replicas, in at least one correct replica. This correct replica must have sent a \fastcommit message and signed the message \noleaderlock, a contradiction since correct replicas never both send \fastcommit and sign the message \noleaderlock.
\end{proof}

\begin{lemma}\label{lem:leadercommitprep2}
If there exists a \leadercommitname, $C$, then any \preptwocertname, $P$, from any proposer in view $0$ must have $P.dig=C.dig$   
\end{lemma}

\begin{proof}
Suppose for the sake of contradiction, $P.dig\neq C.dig$. There are two cases: \one $P$ is a tuple consisting of a \slowcertname and a \nolockcert{} (upgraded) or \two $P$ consists of $n-f$ \recoveryprepvote messages. 

\par \textbf{Case \one.} The existence of $C$ implies that at least $n-2f$ correct replicas sent a \fastcommit message. A correct replica only sends a \fastcommit message if it received a \lockcertname, thus a \lockcertname exists. By lemma \ref{lem:no-lock-cert}, there cannot exist a \lockcertname since $P$ contains a \nolockcert, a contradiction.

\par \textbf{Case \two.} By lemma \ref{lem:commit-no-commit} there cannot exist a \nocommitcert{} and by lemma \ref{lem:leader-lock-commit}, any \lockcertname must have the same digest as a \leadercommitname. However, the existence of $P$ implies that at least $n-2f$ correct replicas sent a \recoveryprepvote message for $P.dig\neq C.dig$. A correct replica will only send a \recoveryprepvote message for $P.dig\neq C.dig$ if the \recoveryprepare message contains a \nocommitcert, a contradiction.
\end{proof}

We say an \elected certificate in view $v$ is a certificate originating from the randomly elected proposer in view $v$. We say a \noelectcertname{} is $n-f$ \noelectcert messages.

\begin{lemma}\label{lem:endorsed-safety}
If there exists an \elected \confirmcertname in view $v$, then there cannot exist a \noelectcertname{} in view $v+1$.
\end{lemma}

\begin{proof}
Suppose for the sake of contradiction there exists both an \elected \confirmcertname in view $v$ and a \noelectcertname{} in view $v+1$. The existence of an \elected \confirmcertname, means that at least $n-2f$ correct replicas sent a \recoveryconfirmvote message for an \elected \preptwocertname. However, the existence of a \noelectcertname{} implies that at least $n-f$ replicas signed \noelectcert. $n-2f$ correct replicas intersects a quorum of $n-f$ replicas in at least one correct replica. This means that a correct replica sent both a \recoveryconfirmvote message indicating it received an \elected \preptwocertname in view $v$, and also signed \noelectcert, a contradiction, since correct replicas never do both.
\end{proof}

\begin{lemma}\label{lem:commit-adopt}
If a correct replica commits $dig$ in view $v$ (recovery path), then any \preptwocertname, $P$, in view $v'>v$ must have $P.dig=dig$
\end{lemma}

\begin{proof}
We prove the lemma by induction.

\par \textbf{Base Case: $v'=v+1$.} Suppose for the sake of contradiction $P.dig\neq dig$. The existence of a \preptwocertname in view $v+1$ implies that at least $n-2f$ correct replicas sent a \recoveryprepvote message for a \recoveryprepare message containing either \one an \elected \preptwocertname from view $v$ or \two a non-\elected \confirmcertname and a \noelectcertname{} for view $v+1$. For case \one, by lemmas \ref{lem:prepare-confirm-equal} and \ref{lem:prep2-confirm-equal} any \elected \preptwocertname in view $v$ must have a digest of $dig$, a contradiction since the \preptwocertname must also be for digest $dig$. For case \two, by lemma \ref{lem:endorsed-safety}, a \noelectcertname{} cannot exist, a contradiction.

\par \textbf{Induction Step.} Assume the lemma holds for all $v'-1$, now consider view $v'$. Suppose for the sake of contradiction $P.dig\neq dig$. The existence of a \preptwocertname in view $v'$ implies that at least $n-2f$ correct replicas sent a \recoveryprepvote message for a \recoveryprepare message containing either \one an \elected \preptwocertname from view $v'-1$ or \two a non-\elected \confirmcertname from view $v'-1$ and a \noelectcertname{} for view $v'$. For case \one, by the induction assumption any \preptwocertname in view $v'-1$ must have $P.dig=dig$, a contradiction. For case \two, by the induction assumption and lemma \ref{lem:prepare-confirm-equal} any \confirmcertname in view $v'-1$ must have digest $=dig$, a contradiction.
\end{proof}

\begin{lemma}\label{lem:commit-prepare}
If there exists a \leadercommitname for digest $dig$, then any \preptwocertname, $P$, in view $v'>0$ must have $P.dig=dig$.
\end{lemma}

\begin{proof}
We prove the lemma by induction.

\par \textbf{Base Case: $v'=1$.} Suppose for the sake of contradiction $P.dig\neq dig$. The existence of a \preptwocertname in view $1$ implies that at least $n-2f$ correct replicas sent a \recoveryprepvote message for a \recoveryprepare message containing either \one an \elected \preptwocertname from view $0$ or \two a non-\elected \confirmcertname from view $0$ and a \noelectcertname{}. By lemma \ref{lem:leadercommitprep2} any \preptwocertname (including \elected) in view $0$ must have digest $=dig$. By lemma \ref{lem:prep2-confirm-equal} any \confirmcertname must have digest $=dig$. Thus a correct replica sent a \recoveryprepvote for neither \one or \two, a contradiction.

\par \textbf{Induction Step.} Assume the lemma holds for all $v'-1$, now consider view $v'$. Suppose for the sake of contradiction $P.dig\neq dig$. The existence of a \preptwocertname in view $v'$ implies that at least $n-2f$ correct replicas sent a \recoveryprepvote message for a \recoveryprepare message containing either \one an \elected \preptwocertname from view $v'-1$ or \two a non-\elected \confirmcertname from view $v'-1$ and a \noelectcertname{} for view $v'$. For case \one, by the induction assumption any \preptwocertname in view $v'-1$ must have $P.dig=dig$, a contradiction, since the \recoveryprepare must adopt the same digest $dig$. For case \two, by the induction assumption and lemma \ref{lem:prepare-confirm-equal} any \confirmcertname in view $v'-1$ must have digest $=dig$, a contradiction, since the \recoveryprepare must adopt the same digest $dig$.
\end{proof}

\begin{theorem}[Safety]\label{thm:safety}
No two correct replicas commit different values.    
\end{theorem}

\begin{proof}
Suppose for the sake of contradiction there exists two correct replicas, $R$ and $R'$, which commit $dig$ and $dig'$, where $dig\neq dig'$ respectively. There are three cases: \one $R$ and $R'$ both committed from receiving a \leadercommitname, \two $R$ and $R'$ both committed from receiving an \elected \confirmcertname, and \three WLOG $R$ committed from receiving a \leadercommitname while $R'$ committed from receiving an \elected \confirmcertname.

For case \one, by lemma \ref{lem:leader-commit} any \leadercommitname must be for the same value, a contradiction. For case \two, let view $v$ and view $v'$ be the views in which $R$ and $R'$ committed respectively. If $v=v'$ and $v=0$, then by lemma \ref{lem:confirm-non-equiv}, $dig=dig'$, a contradiction. If $v=v'$ but $v\neq 0$ then by lemmas~\ref{lem:prepare-non-equiv} and~\ref{lem:prepare-confirm-equal}, $dig=dig'$. Otherwise, WLOG let $v<v'$. By lemma \ref{lem:commit-adopt}, any \preptwocertname in view $v'$ must have digest $=dig$, and by lemma \ref{lem:prepare-confirm-equal} any \elected \confirmcertname must have digest $=dig$, a contradiction. For case \three, let view $v'$ be the view in which $R'$ committed. If $v'>0$ by lemma \ref{lem:commit-prepare}, any \preptwocertname certificate in view $v'$ must have digest $=dig$, and by lemma \ref{lem:prepare-confirm-equal} any \elected \confirmcertname in view $v'$ must also have digest $=dig$, a contradiction. Otherwise, if $v'=0$, the by lemmas~\ref{lem:leadercommitprep2} and~\ref{lem:prep2-confirm-equal} have digest $=dig$, a contradiction.
\end{proof}

\subsection{Liveness}

The elected lane in each view is produced by a common coin. For any set $S$ of
lanes fixed before the election output is known, $\Pr[\elected(v)\in S]=|S|/n$.
In particular, if a replica has \confirmcertname certificates for $n-f$
distinct lanes before the elected lane is known, then the elected lane is among
those lanes with probability $(n-f)/n>2/3$.

Termination certificates are globally verifiable. If a correct replica terminates after receiving either a \leadercommitname or an elected \confirmcertname, then before terminating it forwards that certificate to all replicas. Any correct replica that receives a valid termination certificate terminates, regardless of its current view. We also assume eventual delivery of messages between correct replicas.

\begin{lemma}\label{lem:higher-views}
If no correct replica ever terminates and a correct replica, $R$ is in view $v$, then all correct replicas will eventually advance to view $v+1$.
\end{lemma}

\begin{proof}
We prove the lemma by induction on $v$.

\par \textbf{Base Case: $v=0$.} Suppose for the sake of contradiction some correct replica never enters view $1$. All correct replicas must have entered view $0$ and started the race. Since there are at least $n-f$ correct replicas, and correct replicas only propose valid values, eventually all correct replicas will form a \ulostcert. This will trigger sending a \status message, so all correct replicas will eventually receive $n-f$ \status messages. After receiving $n-f$ \status messages, a correct replica either obtains a
\nolockcert{} and skips \recoveryprepphase{}, or enters
\recoveryprepphase{} and eventually receives at least $n-f$
\recoveryprepvote messages. It then proceeds to persistence phase. It will eventually receive $n-f$ \recoveryconfirmvote messages and complete the persistence phase. All correct replicas will eventually complete the persistence phase, so all correct replicas will eventually receive $n-f$ \confirmcertname and send an \textsc{Elect-Lane} message. Thus every correct replica eventually receives at least $n-f\ge 2f+1$
\textsc{Elect-Lane} messages. Since no correct replica ever terminates, every
correct replica advances to view $1$, a contradiction.

\par \textbf{Induction Step.}
Assume the lemma holds for all views up to $v-1$, and consider view $v$.
Since $R$ is in view $v$, it previously entered view $v-1$. By the induction
hypothesis applied to view $v-1$, all correct replicas eventually advance to
view $v$. All correct replicas will send a \viewchange message at the start of view $v$ and eventually receive $n-f$ \viewchange messages. All correct replicas will then eventually complete the \recoveryprepphase{} and persistence phases in view $v$ since there are at least $n-f$ correct replicas. All correct replicas will thus receive $n-f$ \confirmcertname and send an \textsc{Elect-Lane} message. Thus every correct replica eventually receives at least $n-f\ge 2f+1$
\textsc{Elect-Lane} messages. Since no correct replica ever terminates, every
correct replica advances to view $v+1$.
\end{proof}

\begin{lemma}\label{lem:constant-prob}
If no correct replica terminates before view $v$ and some correct replica enters
view $v$, then some correct replica terminates in view $v$ with probability at
least $(n-f)/n > 2/3$.
\end{lemma}

\begin{proof}
If some correct replica terminates in view $v$ before the election step, then
the claim already holds. Otherwise, no correct replica terminates before the
election step of view $v$.

By the same progress argument as in Lemma~\ref{lem:higher-views}, unless some
correct replica terminates earlier in view $v$, correct replicas eventually
complete the required phases of view $v$. Hence, some correct replica
$R$ eventually receives $n-f$ \confirmcertname certificates for $n-f$ distinct
lanes before the elected lane is known. Let $S$ be this set of lanes.

By the common-coin assumption, $\Pr[\elected(v)\in S]=(n-f)/n > 2/3$. If
$\elected(v)\in S$, then $R$ has a valid \confirmcertname for the elected lane,
so $R$ terminates in view $v$. Therefore, some correct replica terminates in
view $v$ with probability at least $(n-f)/n > 2/3$.
\end{proof}

\begin{lemma}\label{lem:totality}
If a correct replica, $R$, terminates in view $v$, then all correct replicas eventually terminate.
\end{lemma}

\begin{proof}
Suppose for contradiction, a correct replica $R'$ never terminates. If $R$ received a \leadercommitname it forwards it to all correct replicas before terminating. $R'$ will thus receive it and also terminate, a contradiction. If $R$ received an \elected \confirmcertname in view $v$, it will also forward it to all correct replicas before terminating. $R'$ will thus receive it and also terminate, a contradiction.     
\end{proof}

\begin{theorem}[Liveness]\label{thm:liveness}
\sys{} terminates with probability $1$.
\end{theorem}

\begin{proof}
By Lemma~\ref{lem:higher-views}, if no correct replica terminates, then correct
replicas keep entering higher views. Let $p=(n-f)/n > 2/3$. By
Lemma~\ref{lem:constant-prob}, conditioned on no correct replica terminating
before a view, if no correct replica has terminated before that view, some correct replica
terminates in that view with probability at least $p$. Therefore, by a geometric tail bound, the probability that no correct
replica terminates in the first $k$ views is at most $(1-p)^k$, which converges
to $0$ as $k\to\infty$. Hence some correct replica terminates with probability
$1$. By Lemma~\ref{lem:totality}, once one correct replica terminates, all
correct replicas eventually terminate. Therefore, \sys{} terminates with
probability $1$.
\end{proof}

%% file: Tech-Discussion/tech-discussion.tex
\section{Additional Technical Discussion}

By default, the race gives the leader a one-message-delay advantage over the other replicas: the leader can commit in two message delays, whereas non-leader replicas require three. In some deployments, such as LAN environments, this advantage may be insufficient because communication delays between replicas are extremely small. As a result, a non-slow leader may be mistakenly classified as slow, causing replicas to trigger recovery more often than necessary.

To increase the leader's advantage, \sys{} can insert a configurable number of dummy phases before non-leader replicas enter the race. In each dummy phase, replicas send a message containing the phase number and wait for $n-f$ matching messages before advancing to the next phase. Alternatively, \sys{} can introduce a time-based delay, similar to a hedging delay, to further bias the race in favor of the leader. Although dummy phases or time-based delays may may seem at odds with our productivity requirement, replicas still use this time to generate non-equivocation certificates for recovery. By contrast, traditional timeouts and hedging delays do not perform useful protocol work while waiting.

Dummy phases are particularly useful when motorizing \sys. Since coverage (new data layer proposals) to start a slot may be achieved at different times across replicas, the leader may begin a slot later than other replicas, and therefore be detected as slow even when it is not. To mitigate this, replicas can forward their proposals to the leader, allowing the leader to satisfy coverage as well. Because this forwarding step takes one message delay, \sys{} can add one corresponding dummy phase to compensate for the added delay.